\documentclass[sigplan,screen]{acmart}
\renewcommand\footnotetextcopyrightpermission[1]{} 

\def\fullversion

\copyrightyear{2026}
\acmYear{2026}
\setcopyright{cc}
\setcctype{by}
\acmConference[PPoPP '26]{Proceedings of the 31st ACM SIGPLAN Annual Symposium on Principles and Practice of Parallel Programming}{January 31-February 04, 2026}{Sydney, NSW, Australia}
\acmBooktitle{Proceedings of the 31st ACM SIGPLAN Annual Symposium on Principles and Practice of Parallel Programming (PPoPP '26), January 31-February 04, 2026, Sydney, NSW, Australia}
\acmPrice{}
\acmDOI{10.1145/3774934.3786412}
\acmISBN{979-8-4007-2310-0/2026/01}

\usepackage{hyperref}
\usepackage{hyperxmp}
\usepackage{booktabs}   
\usepackage{subcaption} 
\usepackage{caption}
\usepackage{colortbl}
\usepackage{bigstrut}
\usepackage{microtype}
\usepackage{tabularx}
\usepackage{multirow}
\usepackage{multicol}
\usepackage{booktabs}
\usepackage{rotating}
\usepackage{lipsum}
\usepackage{balance}
\usepackage{mfirstuc}
\usepackage{titlecaps}
\usepackage{listings}
\usepackage{makecell}
\usepackage{float}
\usepackage[rightcaption]{sidecap}
\usepackage{setspace}
\usepackage{tcolorbox}
\usepackage[normalem]{ulem}

\usepackage{microtype}
\usepackage[utf8]{inputenc}

\useunder{\uline}{\ul}{}

\usepackage{xcolor}
\usepackage[font=normalfont,labelfont=bf,skip=2pt]{caption}

\newcommand{\oref}[1]{\hyperref[#1]{#1}}

\usepackage{xr}
\usepackage{etoolbox}
\newcommand{\ifconference}[1]{{{\ifx\fullversion\undefined{#1}\fi}\xspace}}
\newcommand{\iffullversion}[1]{{{\ifx\conference\undefined{#1}\fi}\xspace}}

\usepackage{graphicx}  
\usepackage{lipsum}  
\newcommand{\hide}[1]{} 
\usepackage{xspace}
\usepackage{textcomp}
\usepackage{comment} 
\usepackage{verbatim}
\usepackage{url}

\usepackage{pifont}

\usepackage[shortlabels]{enumitem}

\setlist{topsep=0.3em,itemsep=0.2em,parsep=0.1em,leftmargin=*}

\usepackage{float}
\usepackage[font={small},aboveskip=0em, belowskip=0em]{caption}

\usepackage[labelfont=bf,list=true,skip=0em]{subcaption}
\usepackage[rightcaption]{sidecap}

\usepackage{wrapfig}

\usepackage{array}
\newcolumntype{L}[1]{>{\raggedright\let\newline\\\arraybackslash\hspace{0pt}}m{#1}}
\newcolumntype{C}[1]{>{\centering\let\newline\\\arraybackslash\hspace{0pt}}m{#1}}
\newcolumntype{R}[1]{>{\raggedleft\let\newline\\\arraybackslash\hspace{0pt}}m{#1}}
\newcolumntype{B}{>{\bf}c}
\usepackage{rotating}

\usepackage{booktabs} 
\usepackage{multicol,multirow}
\usepackage{longtable} 
\usepackage{supertabular} 
\usepackage{colortbl}
\usepackage{bigstrut}

\usepackage{titlesec}
\titleformat{\subsection}{\normalfont\large\bfseries}{\thesubsection}{1em}{}

\titlespacing{\section}{0pt}{0.3em}{0.2em} 
\titlespacing{\subsection}{0pt}{0.3em}{0.2em} 
\titlespacing{\subsubsection}{0pt}{0.1em}{1em} 
\newcommand{\mysubsubsection}[1]{\underline{\thesubsubsection~~~#1.}}
\titleformat{\subsubsection}[runin]
{\normalfont\normalsize\bfseries}{}{0em}{\mysubsubsection}

\newcommand{\myparagraph}[1]{\vspace{.03in}\noindent {\boldmath\bf #1\unboldmath}}

\newcommand{\rmv}[1]{\textcolor{red}{\sout{}}}

\usepackage[ruled,lined,linesnumbered,noend]{algorithm2e}
\usepackage[noend]{algpseudocode}

\makeatletter
\patchcmd{\@algocf@start}
{-1.5em}
{0pt}
{}{}
\newcommand{\nosemic}{\renewcommand{\@endalgocfline}{\relax}}
\newcommand{\dosemic}{\renewcommand{\@endalgocfline}{\algocf@endline}}

\SetSideCommentLeft
\SetKwInput{notations}{Notations}
\SetKwInput{notes}{Notes}
\SetKwInput{maintains}{Maintains}

\SetKwProg{myfunc}{Function}{}{}
\SetKwFor{parForEach}{ParallelForEach}{do}{endfor}
\SetKwFor{Justrepeat}{Repeat}{}{}

\definecolor{dpcol}{RGB}{0,160,240}

\newcommand{\yan}[1]{{\color{violet}{\bf Yan:} #1}}

\newcommand{\ziyang}[1]{{\color{cyan}{\bf Ziyang:} #1}}

\SetCommentSty{mycommfont}

\usepackage{mdframed}
\definecolor{framelinecolor}{RGB}{68,114,196}
\mdfdefinestyle{mystyle}{linecolor=framelinecolor,innertopmargin=1pt,innerbottommargin=2pt,backgroundcolor=gray!20,skipabove=2pt,skipbelow=0pt}
\mdfdefinestyle{densestyle}{linecolor=framelinecolor,innertopmargin=0,innerbottommargin=0,leftmargin=0,rightmargin=0,backgroundcolor=gray!20}
\mdfdefinestyle{compactcode}{linecolor=framelinecolor,innertopmargin=1pt,innerbottommargin=1pt,backgroundcolor=gray!20,skipabove=0pt,skipbelow=0pt,leftmargin=0,rightmargin=0}

\usepackage{framed}

\usepackage{listings}

\newdimen\zzsize
\newdimen\kwsize
\newcommand{\basicstyle}{\fontsize{\zzsize}{1\zzsize}\ttfamily}
\newcommand{\keywordstyle}{\fontsize{\kwsize}{1\kwsize}\ttfamily\bf}

\newdimen\zzlstwidth
\usepackage{cleveref}
\crefname{appendix}{Appendix}{Appendix}
\crefname{theorem}{Thm.}{Thm.}
\crefname{lemma}{Lem.}{Lem.}
\crefname{corollary}{Col.}{Col.}
\crefname{table}{Tab.}{Tab.}
\crefname{algorithm}{Alg.}{Alg.}
\crefname{figure}{Fig.}{Fig.}
\crefname{fact}{Fact}{Fact}
\Crefname{table}{Tab.}{Tab.}
\crefname{problem}{Problem}{Problem}

\usepackage{tikz} 

\newcommand*{\rom}[1]{\expandafter\@slowromancap\romannumeral #1@}

\newcommand{\treename}[1]{\textsf{#1}}
\newcommand{\cpamtree}{\treename{SPaC-tree}}
\newcommand{\spactree}{\treename{SPaC-tree}}
\newcommand{\hilberttree}{\treename{SPaC-H-tree}}
\newcommand{\mortontree}{\treename{SPaC-Z-tree}}
\newcommand{\porth}{\treename{P-Orth tree}}
\newcommand{\pspt}{$\Psi$-Lib}
\newcommand{\ourlib}{\pspt}
\newcommand{\pactree}{\treename{PaC-tree}}
\newcommand{\pkdtree}{\treename{Pkd-tree}}

\newcommand{\kdtree}{$k$d-tree\xspace}

\newcommand{\orth}{Orth-tree\xspace}
\newcommand{\orthtree}{\orth}
\newcommand{\quadtree}{Quad-tree\xspace}
\newcommand{\octtree}{Oct-tree\xspace}
\newcommand{\ptree}{P-tree\xspace}
\newcommand{\btree}{\textsf{B-tree}\xspace}

\newcommand{\ourorth}{\porth\xspace}
\newcommand{\ourcpam}{\cpamtree\xspace}
\newcommand{\ourcpamh}{\hilberttree\xspace}
\newcommand{\ourcpamz}{\mortontree\xspace}

\newcommand{\ids}{\textsf{InD}\xspace}
\newcommand{\ods}{\textsf{OOD}\xspace}
\newcommand{\pkd}{\pkdtree\xspace}
\newcommand{\zdtree}{\treename{Zd-tree}\xspace}
\newcommand{\osm}{\texttt{OSM}\xspace}
\newcommand{\geolife}{\texttt{GL}\xspace}
\newcommand{\cosmo}{\texttt{CM}\xspace}
\newcommand{\bhltree}{\mbox{\textsf{BHL-tree}}\xspace}
\newcommand{\logtree}{\mbox{\textsf{Log-tree}}\xspace}

\newcommand{\rtree}{\textsf{R-tree}\xspace}
\newcommand{\boost}{\textsf{Boost}\xspace}
\newcommand{\cpam}{\textsf{CPAM}\xspace}
\newcommand{\cpamh}{\textsf{CPAM-H}\xspace}
\newcommand{\cpamz}{\textsf{CPAM-Z}\xspace}
\newcommand{\zd}{\textsf{Zd-tree}\xspace}

\newcommand{\zcurve}{Morton-curve\xspace}
\newcommand{\hcurve}{Hilbert-curve\xspace}
\newcommand{\balpara}{\alpha}
\newcommand{\varden}{\texttt{Varden}\xspace}
\newcommand{\sweep}{\texttt{Sweepline}\xspace}

\newcommand{\uniform}{\texttt{Uniform}\xspace}

\newcommand{\orthBuild}{\textsc{BuildOrth}}
\newcommand{\sortSfc}{\textsc{HybridSort}}
\newcommand{\cpamBuild}{\textsc{BuildSPaCTree}}
\newcommand{\cpamBuildFromSorted}{\textsc{BuildSorted}}

\newcommand{\orthBatchInsert}{\textsc{BatchInsertOrth}}
\newcommand{\cpamBatchInsert}{\textsc{PtreeBatchInsert}}
\newcommand{\cpamBatchInsertFromSorted}{\textsc{InsertSorted}}

\newcommand{\leafwrap}{\phi}
\newcommand{\skheight}{\lambda}

\newcommand{\knn}{$k$-NN}

\newcommand{\os}{\sigma\xspace}

\SetKw{MIN}{min}
\SetKw{MAX}{max}
\SetKw{OR}{or}
\SetKw{AND}{and}

\newcommand{\emp}[1]{\emph{\textbf{\boldmath #1\unboldmath}}} 
\newcommand{\fname}[1]{\textsc{#1}} 

\newcommand{\ip}[2]{\langle#1,#2\rangle}
\newcommand{\bdita}[1]{{\boldmath\textbf{\textit{#1}}\unboldmath}}

\newcommand{\intersection}{\fname{Intersection}}
\newcommand{\join}{\fname{Join}}
\newcommand{\node}{\fname{Node}}

\newcommand{\leftJoin}{\fname{LeftJoin}}
\newcommand{\rightJoin}{\fname{RightJoin}}
\newcommand{\expose}{\fname{Expose}}
\newcommand{\union}{\fname{Union}}

\newcommand{\intersect}{\intersection}
\newcommand{\diff}{\fname{Diff}}

\newcommand{\R}{\mathbb{R}}

\newcommand{\matht}{\mathcal{T}\xspace}

\newtheorem{theorem}{Theorem}[section]
\newtheorem{lemma}[theorem]{Lemma}

\let \originalleft \left
\let\originalright\right
\renewcommand{\left}{\mathopen{}\mathclose\bgroup\originalleft}
\renewcommand{\right}{\aftergroup\egroup\originalright}

\usepackage{scalerel} 

\newtheoremstyle{exampstyle}
{.5em} 
{1em} 
{\it} 
{.5em} 
{\it \bfseries} 
{.} 
{.5em} 
{} 

\makeatletter
\renewenvironment{proof}[1][\proofname]{\par
  \vspace{-2\topsep}
  \pushQED{\qed}%
  \normalfont
  \topsep0pt \partopsep0pt 
  \trivlist
  \item[\hskip\labelsep
              \itshape
              #1\@addpunct{.}]\ignorespaces
}{%
  \popQED\endtrivlist\@endpefalse
}

\crefname{section}{Sec.}{Sec.}
\crefname{table}{Tab.}{Tab.}
\crefname{figure}{Fig.}{Fig.}

\makeatletter
\newcommand{\pushright}[1]{\ifmeasuring@#1\else\omit\hfill$\displaystyle#1$\fi\ignorespaces}
\newcommand{\pushleft}[1]{\ifmeasuring@#1\else\omit$\displaystyle#1$\hfill\fi\ignorespaces}
\makeatother

\begin{document}

\makeatletter
\gdef\@copyrightpermission{
  \begin{minipage}{0.2\columnwidth}
    \href{https://creativecommons.org/licenses/by/4.0/}{\includegraphics[width=0.90\textwidth]{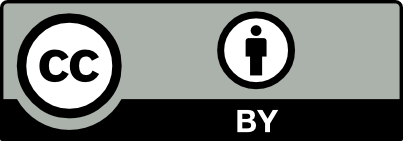}}
  \end{minipage}\hfill
  \begin{minipage}{0.8\columnwidth}
    \href{https://creativecommons.org/licenses/by/4.0/}{This work is
    licensed under a Creative Commons Attribution International 4.0 License.}
  \end{minipage}
}
\makeatother



\settopmatter{authorsperrow=4}
%

\title{Parallel Dynamic Spatial Indexes}

\author{Ziyang Men}
\orcid{0000-0001-7290-690X}
\affiliation{%
  \institution{UC Riverside}
  \city{Riverside}
  \country{USA}
}
\email{ziyang.men@email.ucr.edu}

\author{Bo Huang}
\orcid{0009-0000-9314-1096}
\affiliation{%
  \institution{UC Riverside}
  \city{Riverside}
  \country{USA}
}
\email{bo.huang@email.ucr.edu}

\author{Yan Gu}
\orcid{0000-0002-4392-4022}
\affiliation{%
  \institution{UC Riverside}
  \city{Riverside}
  \country{USA}
}
\email{ygu@cs.ucr.edu}

\author{Yihan Sun}
\orcid{0000-0002-3212-0934}
\affiliation{%
  \institution{UC Riverside}
  \city{Riverside}
  \country{USA}
}
\email{yihans@cs.ucr.edu}

\begin{abstract}
  Maintaining spatial data (points in two or three dimensions) is crucial and has a wide range of applications, such as graphics, GIS, and robotics.
  To support efficient updates and queries on the spatial data,
  many data structures, called \emph{spatial indexes}, have been proposed, e.g., \kdtree{s}, oct/quadtrees (also called \orth{s}), R-trees, and bounding volume hierarchies (BVHs).
  In real-world applications, spatial datasets tend to be highly dynamic, requiring batch updates of points with low latency.
  This calls for efficient parallel batch updates on spatial indexes.
  Unfortunately, there is very little work that achieves this.

  In this paper, we systematically study parallel spatial indexes, with a special focus on achieving high-performance update performance for highly dynamic workloads.
  We select two types of spatial indexes that are considered optimized for low-latency updates: \orth{} and R-tree/BVH.
  We propose two data structures: the \porth{}, a parallel \orth{}, and the \cpamtree{} family, a parallel R-tree/BVH.
  Both the \porth{} and the \cpamtree{} deliver superior performance in batch updates compared to existing parallel \kdtree{s} and \orth{s}, while preserving better or competitive query performance relative to their corresponding \orth{} and R-tree counterparts.
  We also present comprehensive experiments comparing the performance of various parallel spatial indexes and share our findings at the end of the paper.
  \let\thefootnote\relax\footnote{To appear: The ACM SIGPLAN Symposium on Principles and Practice of Parallel Programming (PPoPP), 2026. }
  \setcounter{footnote}{0}
\end{abstract}



\fancyhead{} 

\maketitle
\section{Introduction}

Spatial data widely appear in geographic information systems (GIS), spatial databases, computer graphics, robotics and its planning, and many other domains.
Efficiently processing such geometric objects (usually points) in two or three dimensions is of great importance, for both maintenance (construction, insertion, deletion) and queries (range queries, nearest-neighbor queries, etc.).

Given the wide applicability, many well-known data structures (usually called ``spatial indexes'') have been proposed to handle spatial data, such as \kdtree{s}\cite{bentley1975multidimensional}, oct/quadtrees\cite{finkel1974quad} (collectively referred to as \emph{orth-trees}), range trees~\cite{bentley1978decomposable}, R-trees~\cite{guttman1984r}, and bounding volume hierarchies (BVHs)~\cite{akenine2019real}.
Spatial indexes typically organize points as a tree, with each subtree corresponding to a subspace (not necessarily non-overlapping).
The bounding boxes of the subspaces can be used to prune subtrees during queries.
For instance, consider a nearest-neighbor search: when the search reaches a subtree, if its sub-region is farther from the query point than the current nearest neighbor, the subtree can be pruned.
Despite maintaining different invariants, all these trees share the same high-level intuition: skip most of the objects in queries by pruning, leading to efficient query performance.

Real-world applications can involve highly dynamic data, and updates may be latency-sensitive or throughput-sensitive.
For example, in 3D games, moving objects must be reflected quickly to affect lighting and collision detection, whereas GIS applications often ingest high-volume sensor streams where total update throughput is critical.
In both scenarios, updates frequently arrive in batches and must be incorporated into the index promptly.
To handle both updates and queries efficiently, different spatial indexes offer different trade-offs.
Traditionally, \kdtree{s} are considered highly efficient for queries due to their strongest invariant (splitting at object medians), but updates are costly.
\orth{s} offer competitive query performance and faster updates due to their simpler invariant (splitting at spatial medians).
R-trees/BVHs encompass a large family of solutions; they usually provide the simplest and fastest updates but slower queries.

With the ever-growing data volume, \emph{parallelism} becomes essential in designing efficient data structures.
Unfortunately, little work is known on parallel spatial indexes with batch updates.
In the two famous libraries, CGAL~\cite{fabri2009cgal} and Boost~\cite{schaling2011boost}, most spatial indexes are sequential.
The only exception is CGAL's \kdtree{}, but it has known scalability issues~\cite{blelloch2022parallel, men2025parallel}.
Parallel construction for range trees was described in~\cite{sun2018pam}, but it does not support batch updates.
In 2022, Blelloch and Dobson~\cite{blelloch2022parallel} proposed \zdtree{}, the first parallel quadtree.
The idea is to leverage the Morton curve, a space-filling curve (SFC) that maps 2D or 3D points to 1D integers, and use this information to facilitate construction and batch updates.
However, \zdtree{s} are slower than the parallel \kdtree{} (the \pkdtree{}), proposed later~\cite{men2025parallel}, despite better theoretical bounds for updates ($O(\log n)$ vs.\ $O(\log^2 n)$ per updated point).
We believe the main reason is the I/O (cache) optimizations in \pkdtree{}.
More interestingly, despite R-trees/BVHs having the simplest structure, we are aware of little work on parallel batch updates for them.
Indeed, most existing approaches are either based on single insertions/deletions~\cite{schaling2011boost,guttman1984r, sellis1987r, arge2002efficient,bittner2015incremental},
or fully rebuilding the tree upon updates~\cite{gu2013efficient,pantaleoni2010hlbvh,viitanen2018ploctree}.
The only relevant work~\cite{Qi2018theoretically} uses the logarithmic method, which can substantially slow down the query time (see more details in \cref{pre:existingSpatilIndex}).
Hence, it is natural to ask whether \orth{s} and R-trees/BVHs can still leverage their strengths for highly dynamic workloads in the parallel setting.
In particular, we investigate whether they can achieve much faster construction and batch updates (with better theoretical guarantees) than \kdtree{s} in parallel, while preserving query performance as in their sequential counterparts.

\textbf{In this paper, we systematically study parallel spatial indexes, with a special focus on achieving high-throughput updates and good query efficiency in highly dynamic workloads}.
We propose two new (families of) data structures: \textbf{\porth{s}} and the \textbf{\cpamtree{} family}.
We integrate these data structures into a library called the \emph{Parallel Spatial Index Library}~\cite{psiCode}, abbreviated as PSI-Lib or {\pspt}.

We first show our design for a parallel \orth{} called the \porth{}.
Almost all existing \orth{s}~\cite{kim2013scalable, warren1993parallel, blelloch2022parallel, pantaleoni2010hlbvh, lauterbach2009fast} use space-filling curves (SFCs) to accelerate construction and updates.
However, simply computing and sorting the SFC codes of the points already requires several passes of reading and moving all data, which is time-consuming.
In this paper, we present the design of \porth{s} that does not use SFCs.

By definition of orth trees, in $D$ dimensions, 
our main idea is to directly split the space evenly into $2^D$ buckets (subspaces), 
and partition input points into each bucket accordingly. 
To do this, we borrow the idea of the \emph{sieving algorithm} from the \pkdtree{}~\cite{men2025parallel}, which directly reorders the input points and gather those in the same buckets together. Then each bucket is processed in parallel. This allows for I/O-efficient construction and batch update algorithms for \orth{}.

Conceptually, our algorithms are equivalent to integer-sorting SFC codes, but without generating, storing, or using them.
We believe the algorithmic idea is interesting, and refer readers to \cref{sec:orth_tree} for algorithmic details and analysis.

Our next question, then, is whether SFCs are still useful spatial indexes.
As mentioned, SFCs have been used in both \orth{} and R-trees/BVHs~\cite{kim2013scalable, warren1993parallel, blelloch2022parallel, pantaleoni2010hlbvh, lauterbach2009fast, gu2013efficient, qi2020packing, tropf1981multidimensional, kamel1993packing}.
However, we are unaware of any implementations with update performance competitive with \pkdtree{} and \ourorth{s}, mostly due to limited or no parallel support.

In this paper, we propose the \cpamtree{} family, which supports extremely fast updates (as R-trees are supposed to) while maintaining query performance competitive with existing R-trees/BVHs.
To achieve this, our backbone is the \pactree{}\cite{dhulipala2022pac}, a parallel balanced binary search tree.
The key insight of \pactree{s} is to use join-based algorithms\cite{adams1992implementing, adams1993functional, blelloch2016just, sun2018pam} to efficiently rebalance during parallel updates, and to use leaf blocking (maintaining 16--32 objects in each leaf in a flat array) to improve cache locality.
To support spatial queries, a simple approach is to store points using their SFC codes as keys in a \pactree{} and augment each tree node with bounding boxes.
However, this plain adaptation yields poor update speed (up to $3.5\times$ slower than \pkdtree{}; see the columns ``\cpamh'' and ``\cpamz'' in \cref{fig:heatmap}).
We observe that the main bottleneck is maintaining the SFC-induced \emph{total} order over all points in \pactree{s}.
To address this challenge, we carefully redesign the join-based algorithms in \pactree{s} to maintain spatial data under only a \emph{partial order}, i.e., the points in leaves are allowed to be unsorted after insertions/deletions.
We provide more details in \cref{sec:spac}.
We refer to our design as the Spatial \pactree{}, or \spactree{} for short. In \ourlib{}, we adopt both Morton curves (\mortontree{}) and Hilbert curves (\hilberttree{}).

Our \ourorth{s} and \spactree{s} are backed by strong theoretical support.
We show that the update cost per object is $O(\log n)$ for a \spactree{} and $O(\log \Delta)$ for a \ourorth{} ($\Delta$ is the aspect ratio, see \cref{sec:orth_analysis}), which is much stronger than $O(\log^2 n)$ for a \pkdtree.
Our batch updates achieve polylogarithmic span, indicating strong and scalable parallelism.

We tested \ourlib{} on workloads with various input distributions, query distributions, query types, and update patterns. 
We compare \ourlib{} with existing parallel and sequential baselines including \pkdtree{s}, \zdtree{s}, etc.
Our experiments simulates both a static setting and a highly dynamic setting where updates are consecutively applied to an initial tree.
This setting better reflects the capability of each data structure to handle highly dynamic workloads,
especially showcases whether and how the index quality are affected under a progressively evolving dataset.
With our new algorithms, both \porth{} and \cpamtree{} achieved superior construction and update performance, while preserving comparable query performance to regular \orth{} and R-trees.
\porth{} is almost always the fastest on uniformly distributed data in construction and queries, and is close to the best on updates.
\cpamtree{} supports extremely fast parallel batch updates---it can be 2-6 times faster than \pkdtree{s}, and is especially good for skewed distribution of input points, queries, or insertion/deletion orders. With comprehensive experiments, we share our findings in \cref{sec:exp:summary}, and visualize the query-update tradeoff of each parallel spatial index in \cref{fig:summary}.
Our code is available at~\cite{psiCode}.



\hide{
  \section{Introduction}\label{sec:intro}

  Managing spatial data is ...

  One of the common approach to manage the spatial data is partitioning the space into smaller regions, and then storing the data in a hierarchical tree structure.
  Data structures such as the \orth{}~\cite{xxx} (known as \quadtree{} in 2-dimensional space and \octtree{} in 3-dimensional space) is widely used for this purpose.
  Unlike its counterpart \kdtree{}~\cite{xxx}, which partitions the space based on the data distribution using the object median,
  the \orth{} organizes the data hierarchically to mirror a regular grid structure, i.e., divides the space into equal-sized regions.
  Such scheme allows the \orth{} a faster tree construction than \kdtree{} when inputs are uniformly distributed in the space, and enables it to handle the updates more efficiently without sacrificing the query performance, meanwhile provides a more natural way to simulate the regular subdivision in graphic and spatial indexing applications~\cite{xxx}.
  However, we observe a significant gap between the wide usage of \orth{} and the lack of efficient parallel algorithms for it for all the basic operations, including construction, batch updates, and query.
  Existing parallel implementation for \orth{} are either suboptimal in cache efficiency~\cite{xxx} or integrates limited parallisim~\cite{xxx} so that is not scalable to large datasets~\cite{xxx}.
  Actually, the tree construction of these implementation is much slower than the recently reported parallel trees such as the parallel \kdtree{}~\cite{men2025parallel} even when the input is uniform, indicating a significant space for improvement of tree construction algorithms.

  Another common approach to manage the low-dimensional data is to utilizes the space-filling curve (SFC), such as \hcurve{} and \zcurve{}, to map the data into a single dimension in a way that largely preserves spatial locality -- points are close to each other in $k$-dimensional space tends to have one-dimensional keys.
  Such method enables to handle the data using the well-tuned one-dimensional data structures such as \btree{s}, binary search tree and interval trees.
  The contiguous location of points introduces more cache-friendly access pattern, whereas arbitrary tree traversals scatter memory accesses.
  Existing parallel algorithms on SFC-based indexes handles the updates in a poor way:
  some implementations such as \cite{xxx} is static and does not support the batch updates,
  the only dynamic implementation we are aware of is the \zdtree{}~\cite{zdtree},
  where they simply insert and delete the points as is without any balancing scheme, resulting a dramatically performance degradation in the consecutive update scenario.
  Moreover, we notice that there is no implementation supports persistence and Multi-Version Concurrent Control in the dynamic SFC-based indexes, which is a common requirement in many applications, such as~\cite{xxx}.

  In this paper, we overcome above issues by proposing the \ourorth{} (Parallel \orth{}) and \ourcpam{} (Join-based Spatial Partition tree), which are parallel in-memory trees that is efficient in both work and span.

  Spatial data processing plays an important role in a wide range of real-wold applications, including geographic databases, computer graphics, robotic planning, etc. Such data typically consist of points or other geometry objects in two or three dimensions, and often handle tasks such as nearest-neighbor search and range queries.

  To facilitate such applications, many data structures have been proposed to handle spatial indexes, such as kd-trees~\cite{}, oct/quadtrees~\cite{} (also collectively referred to as \emp{orth-trees}), range trees~\cite{}, priority trees~\cite{}, R-trees~\cite{}, etc. Such data structures usually organize the objects in a tree-based structure, and maintains the bounding boxes of each subtree, explicitly or implicitly, to help effectively prune unrelated subtrees in queries.
  Take the nearest neighbor search as an example, when a search reaches a subtree, if the entire bounding box is farther than the current nearest neighbor,
  then the subtree can be skipped in the query, potentially leading to much fewer objects searched in the query.
  For this reason, these spatial data structures are widely used across many applications.

  Despite the fast queries, an inherent challenge for these data structures is the expensive, or complicated process, for handling frequent changes of the dataset. In most of the cases, changes may require an insertion/deletion/update of batch of objects. Reflecting the changes on to the data structure in time (i.e., with low latency) is essential to maintaining a dynamic dataset in the application of interest. For example, in computer graphics, a moving object in a scene may result in a batch of points to change their position every frame.
  As such, the tree either needs to be able to handle rapid batch update in an efficient way~\cite{}, or needs to be reconstructed with minimum delay~\cite{}.
  In these cases, while certain data structures such as kd-trees and range trees typically offers strong performance and/or worst-case bounds in queries, these strong properties also add extra burden in dynamizing or reconstructing the tree.
  Consequently, for highly-dynamic scenarios, practitioners typically favor spatial data structures such as octrees/quadtrees, R-trees, and bounding volume hierarchies (BVHs).
  Although these structures do not provide the strong query performance or worst-case guarantees of kd-trees and range trees, their simplicity and flexibility make them well suited to rapid, massive updates.

  Unlike previous work of zd-trees that uses space filling curves and integer sort for performance engineering, our \porth{} borrows the idea from I/O-efficient parallel kd-tree. In particular, we observe that simply computing and sorting the morton code of all input data already involves several rounds of reading and moving all data, which can already consume much time.
  Instead of presorting them, we borrow the idea of \emph{sieving algorithm} in pkd-tree, achieving I/O-efficient algorithms of construction and batch updates for orth-trees.

  In addition, inspired by zd-tree's idea on using space filling curves, we proposed another data structure, the \cpamtree{} family, which is compatible with any space filling curves, and can be viewed as a special type of R-tree or BVH data structure. Instead of using the morton curve to aid the construction and update of orth-trees, we directly embed the code of the space filling curves as the keys in a parallel search tree. In particular, we adopted the algorithms in \pactree{}~\cite{dhulipala2022pac}, a parallel and I/O-optimized binary search tree.
  More precisely, \pactree{} is based on a weight-balanced tree, and used the join-based algorithm~\cite{adams1992implementing, adams1993functional, blelloch2016just, sun2018pam} for rebalancing in parallel construction, insertion and deletion.
  To make the data structure more cache-friendly, \pactree{} wraps the regions close to the leaves into flat arrays, which greatly reduces the space-overhead and improves locality, enabling better performance. We present more details of this algorithm in \cref{sec:prelim}.
  To support spatial queries, we use the \emph{augmented values} in \pactree{} to maintain the bounding boxes of each subtree.
  However, we note that directly embedding the code as the key for \pactree{} does not directly give competitive performance compared to other spatial data structures such as \pkdtree{s} (up to 4x slower for frequent updates).
  The reason being that, maintaining the full ordering of points based on the space filling curve code is redundant---for the last several levels, the strict ordering often does not significantly help with pruning in queries, but requires a considerable amount of time to maintain. To address this challenge, we relax the ordering in the leaves of \pactree{s}, and carefully redesigned the algorithms such that we allow for minimal changes to make them compatible with the exiting algorithmic framework in \pactree{s}.

}




%
\section{Preliminaries and Related Work}
\label{sec:prelim}

\begin{table}
  \small
  \begin{tabular}{|>{\boldmath}l<{\unboldmath}p{3.2cm}>{\boldmath}l<{\unboldmath}p{2.6cm}|}
    \hline
    $T$         & \multicolumn{3}{l|}{a (sub-)\kdtree, also the set of points in the tree}                                             \\
    $\leafwrap$ & \multicolumn{3}{l|}{leaf wrap threshold (leaf size upper bound)}                                                     \\
    $k$         & \multicolumn{3}{l|}{required number of nearest neighbors in a query}                                                 \\
    $\skheight$ & \multicolumn{3}{l|}{number of levels in a tree sketch}                              \\
    $\matht$    & \multicolumn{3}{l|}{tree skeleton at $T$ with maximum levels $\skheight$}                                            \\
    $P$         & \multicolumn{3}{l|}{input/update point sequence}                              \\
    $P[l,r)$    & \multicolumn{3}{l|}{points in $P$ over the index range $[l, r-1]$}                                               \\
    $T_p$       & \multicolumn{3}{l|}{the pivot point associate with non-leaf node $T$}                                 \\
    $T_\ell$    & left child of $T$                                                                       & $T_r$ & right child of $T$ \\
    $n$         & tree size                                                                               &
    $m$         & batch size                                                                                                           \\
    $D$         & number of dimensions                                                                    &
    $d$         & a certain dimension                                                                                                  \\
    $M$         & \multicolumn{3}{l|}{small memory (cache) size}                                                                \\
    $B$         & \multicolumn{3}{l|}{memory block (cacheline) size}  \\
    \hline
  \end{tabular}
  \caption{\textbf{Notations used in this paper.} }
  \label{tab:notations}
\end{table}

\hide{
  \begin{table}
    \small
    \begin{tabular}{|cp{7.5cm}|}
      \hline
      $M$         & small-memory (cache) size                                                             \\
      $B$         & cacheline size                                                                        \\
      $D$         & number of dimensions                                                                  \\
      $d$         & a certain dimension                                                                   \\
      $k$         & required number of nearest neighbors in a query                                       \\
      $T$         & a (sub-)\kdtree, also the set of points in the tree                                   \\
      $T.lc$      & left child of $T$, so as for $T.rc$                                                   \\
      $\skheight$ & number of levels in tree sketch (i.e., that are built at a time)                      \\
      $\matht$    & shorthand of $\matht_\skheight$, tree skeleton at $T$ with maximum levels $\skheight$ \\
      $P$         & input point set (for updates, $P$ is the batch to be updated)                         \\
      $n$         & tree size                                                                             \\
      $m$         & batch size for batch updates                                                          \\
      $S$         & samples from $P$                                                                      \\
      $s$         & size of the $S$                                                                       \\
      $\os$       & oversampling rate                                                                     \\
      $\balpara$  & balancing parameter                                                                   \\
      $\leafwrap$ & upper bound for size of leaf wrap                                                     \\
      \hline
    \end{tabular}
    \caption{Notations used in this paper. }
    \label{tab:notations}
  \end{table}
}

Throughout the paper, we use $n$ to denote the input size or the tree size.
We use the $\log n$ notation to denote the $\log_2 (n+1)$ logarithm.
We summarize notations used in this paper in \cref{tab:notations}.

\subsection{Computational Models}

We consider the shared-memory multiprocessor setting with the classical fork-join paradigm with binary forking~\cite{BL98,arora2001thread,blelloch2020optimal}.
Each computational thread is a sequential Random Access Machine (RAM) augmented with a fork instruction that spawns two child threads executing in parallel, with the parent thread resuming upon completion of both children.
Parallel for-loops are efficiently simulated through logarithmic levels of forking.
When analyzing algorithms, we use the work-span model, where the work is the total number of operations in the algorithm and the span is the longest dependence chain in the parallel computation.
Using randomized work-stealing schedulers, a computation with work $W$ and span $S$ executes in $W/\rho+O(S)$ time with high probability (in $W$) on $\rho$ processors~\cite{BL98,arora2001thread,gu2022analysis}.

We use the ideal-cache model~\cite{Frigo99} to analyze the I/O cost of our algorithms.
In this model, memory is divided into two levels: a fast memory (cache) of size $M$ and an arbitrarily large slow memory.
The CPU can only access data in the fast memory (at no cost), and data is transferred between the two levels in blocks of size $B$.
Each block transfer incurs unit cost.
The cache is fully associative, and the optimal offline cache replacement policy is used.
The cache complexity of an algorithm is measured by the number of block transfers between the two levels of memory during its execution.

\subsection{Spatial Data}\label{sec:prelim-spatial}
In this paper, we study points in Euclidean space $\mathbb{R}^D$ for $D=2$ or $3$, although the proposed techniques can generalize to shapes and any constant integer $D>1$.

\myparagraph{Queries on Spatial Data.}
To benchmark the quality of spatial indexes, we use standard $k$-NN queries and range queries.
A $k$-nearest neighbor ($k$-NN) query takes a set of points $P$ and a query point $q$ as input, and returns the $k$-closest points to $q$ in $P$.
A range query takes a set of points $P$ and an axis-aligned rectangle subregion~$r$.
The range-count query returns the number of points in $P$ within $r$, and the range-list query returns all points within $r$.

\myparagraph{Spatial Filling Curves.}
A spatial filling curve (SFC) embeds multidimensional points into a one-dimensional sequence.
In \ourlib{}, we use \emph{Z-curve (\zcurve{})} and \emph{\hcurve{}}, illustrated in \cref{fig:SFC}.
Both of them encode each point as an integer, which determines the point's order along the curve.
For integer coordinates, both Hilbert- and Z-curve can be computed in a constant time.
SFCs are widely used to facilitate spatial indexes~\cite{blelloch2022parallel,pantaleoni2010hlbvh,lauterbach2009fast,gu2013efficient,qi2020packing,kamel1993packing}.

\subsection{Existing Commonly-Used Spatial Indexes\label{pre:existingSpatilIndex}}
\vspace{-.2em}
\myparagraph{Space-Partitioning Trees: \orth{s} and \kdtree{s}.}
In space-partitioning trees, each node represents a subspace.
All of its children form a non-overlapping partition of that subspace, usually by axis-aligned partition hyperplanes, i.e., a splitting dimension $d$ and a coordinate $x$.
Space-partitioning trees thus differ in how they select partition hyperplanes.

As typical examples, a \kdtree{}~\cite{bentley1975multidimensional} chooses the median coordinate in the splitting dimension across all points, and thus always yields a balanced partition into two subtrees.
An orth-tree in $D$ dimensions partitions the space into $2^D$ subspaces evenly using the midpoint in each dimension (and is therefore a $2^D$-ary tree).
Specifically, an orth-tree is called a quadtree~\cite{finkel1974quad} in 2D and an octree in 3D~\cite{jackins1980oct}.

There are parallel versions of both \kdtree{s} and \orthtree{s}. Blelloch et al.~\cite{blelloch2022parallel} proposed a parallel \orth{} called \zdtree{}, which uses Morton curve to facilitate construction and updates.
Yesantharao et al.~\cite{yesantharao2021parallel} proposed two parallel \kdtree{s}, \bhltree{} and \logtree{}.
Only \logtree{s} support efficient parallel batch updates, using the \emph{logarithmic method}, i.e., it maintains $O(\log n)$ trees with sizes 1, 2, ... $n/2$, such that a batch update can be broken down into at most $O(\log n)$ tree reconstructions.
However, this method can greatly slow down queries~\cite{men2025parallel}.
A recent work proposed the \pkdtree{}~\cite{men2025parallel} that avoids logarithmic method, and achieves optimal work and cache complexity for parallel construction and batch updates.
The underlying idea is to use sampling to approximate the object median, together with the \textit{sieving} algorithm to partition points in an I/O-efficient manner.
Our \porth{s} also borrow this idea; see \cref{sec:orth_tree} for details.
However, the \pkdtree{} requires $O(m\log^2 n)$ work to update a batch of size $m$.
We will show how \ourlib{} achieves better bounds.

\myparagraph{Object-Partitioning Trees: R-Trees/BVHs.}
In the object- partitioning trees, the objects (points) in each (sub)tree are partitioned into disjoint subsets, and each subset corresponds to a child node and is built recursively.
Each tree node typically stores a bounding box (or a bounding volume in 3D) that is the smallest enclosing axis-aligned region of all objects in its subtree.
Though named differently---\rtree{s} in databases and usually in 2D (``R'' for rectangle), and bounding volume hierarchies (BVHs) in graphics and usually in 3D (``V'' for volume)---they share the same underlying concept.
For simplicity, we use the term ``R-tree'' to refer to the general idea of object-partitioning trees.
They can be either binary~\cite{gu2013efficient,wald2008fast,bittner2015incremental} or have a larger branching factor~\cite{schaling2011boost,guttman1984r,kamel1993packing,gu2015ray}, and can be built either offline~\cite{schaling2011boost,gu2013efficient,pantaleoni2010hlbvh,viitanen2018ploctree} or incrementally (thus supporting updates)~\cite{schaling2011boost,guttman1984r,sellis1987r,arge2002efficient,bittner2015incremental,kamel1993packing}.

To our knowledge, the only parallel R-tree with batch updates is by Qi et al.~\cite{qi2020packing}, which uses the logarithmic method. However, as noted, the logarithmic method significantly slows down queries and is therefore non-ideal. There also exist lock-based concurrent R-trees~~\cite{chenhhc97concurrent,ngk94rlink}.


\hide{
  The \ptree{}~\cite{xxx} is a parallel binary search tree for representing sequences, ordered sets, ordered maps and augmented maps (as defined in~\cite{xxx}).
  The \ptree{} is a purely functional data structure, where each update operation creates a new version of the tree without modifying the original one by using the path copying, so that achieving functional updates in a low cost, e.g., $O(\log n)$ work per point.
  A \ptree{} can be constructed in $O(n\log n)$ work and $O(\log n)$ span.
  For batch updates,
  the key premitive in \ptree{} is the \join{} operation, which takes two trees $L$ with size $n$, $R$ with size $m$ and a key-value pair $(k,v)$, and returns a new tree that contains all elements in $L$, $R$ and $(k,v)$ in $O(\log n-\log m)$ work and span.
  By using the \join{} operation, \ptree{} can efficiently implement various batch operations, such as \union{}, \intersect{} and \diff{} with optimal $O(m\log(\frac{n}{m}+1))$ work and $O(\log n\log m)$ span.
  \myparagraph{The \pactree{}.}
}
\hide{
  The \ptree{}~\cite{sun2018pam, blelloch2016just, blelloch2022joinable} can be considered a standard parallel binary search tree (BST) that supports a full interface for construction, single and batch updates, set operations, and many 1D queries.
  The \ptree{} uses either the AVL tree, the red-black tree, the treap, or the weight-balance tree as the underlying structure, and adopts a so-called ``\join{}-based framework'' in a divide-and-conquer manner for achieving high update parallelism.
  Any update in the \ptree{} is immutable, which generates a new version and leads to high space usage.
  The \pactree~\cite{dhulipala2022pac} is built based on the \ptree{} with the leaf wrapping technique.
  Any subtree whose size is no more than a given threshold will be flattened into a compressed leaf node with all points stored in an array.
  In this way, the auxiliary space required for the tree structure becomes negligible compared to the size of the spatial data.
  To support the leaf wrapping technique and achieve both theoretical and practical efficiency, the ``\join{}-based framework'' needs to be redesigned for the \pactree{} to handle the transformation between compressed nodes and normal nodes.
}

\begin{figure}
  \centering
  \includegraphics[width=\columnwidth]{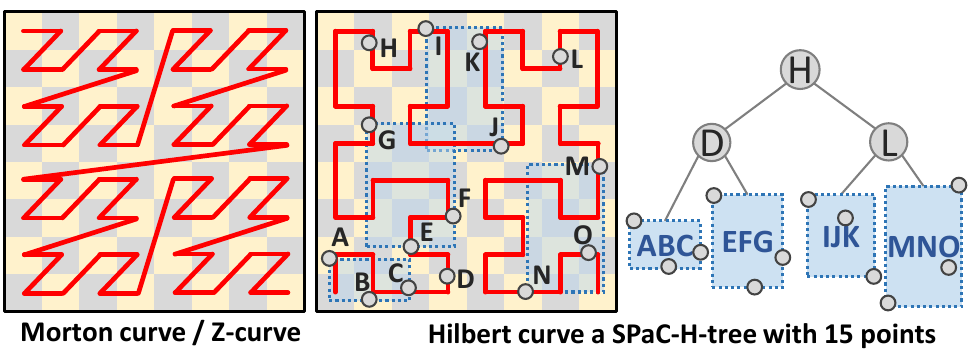}
  \caption{Space-filling curves and an example of a \spactree{} with 15 points and size-3 leaf wrapping.  Each leaf in this case has 3 points and its bounding box marked in blue.}\label{fig:SFC}
\end{figure}

\hide{
  An \orth{} is a hierarchical partitioning of $\mathbb{R}^D$ space using axis-aligned hyperplanes, which is known as the quad-tree~\cite{finkel1974quad} for $D=2$ and oct-tree~\cite{jackins1980oct} for $D=3$.
  Each interior node~$T$ in an \orth{} is associated $D$ orthogonal hyperlanes.
  The $i$-th hyperland $\ip{d_i}{x_i}$ splits the space along the \textit{spatial median} of the coordinates $x_i$ in the dimension $d_i$.
  We refer to such $D$ hyperlanes as the \emph{splitter} of $T$.
  In this case, every interior node partition the space into $2^D$ subspaces,
  and each subspace is associated with a child of $T$.

  \yan{Introduce PKD-trees.}

  \myparagraph{\kdtree{s} and \pkdtree{s}.} Similar to the \orth{}, \kdtree{s}~\cite{bentley1975multidimensional} are binary space partition trees that split the space along the \textit{object median}.
  The \pkdtree{}~\cite{men2025parallel} is the state-of-the-art parallel \kdtree{}
  that achieves optimal work and I/O complexity for both tree construction and batch updates, meawhile provides high parallelism.
  The underlying idea is to use sampling to approximate the object median, together with a I/O efficient \textit{sieving} algorithm to partition points in an I/O efficient manner.
  The \pkdtree{} maintains a weight-balanced tree during updates and use the partial rebuild~\cite{overmars1981maintenance} to restore the imbalanced sub-tree.

  \myparagraph{\rtree{s}/BVHs.}
  Unlike \kdtree{s} and \orth{s} that are based on spatial partition, \rtree{s} or bounding volume hierarchies (BVHs) are \emph{objection-partition} trees.
  Objects (points) in each (sub)tree is partitioned into several disjoint sets, and each set corresponds to a child node and is built recursively.
  To facilitate spatial queries, the each tree node augments a bounding box (or bounding volume in 3D) that is the smallest enclosing axis-aligned region of all objects in this subtree.
  The two names are used in two communities: \rtree{s} in the database community and BVHs for computer graphics, and they share the same idea.
  They can either be binary~\cite{gu2013efficient,wald2008fast,bittner2015incremental} or have a larger branching factor~\cite{schaling2011boost,guttman1984r,kamel1993packing,gu2015ray}, and can be built either offline~\cite{schaling2011boost,gu2013efficient,pantaleoni2010hlbvh,viitanen2018ploctree} or incrementally (thus supporting updates)~\cite{schaling2011boost,guttman1984r, sellis1987r, arge2002efficient,bittner2015incremental,kamel1993packing}.
  In this paper, we will mainly use the term R-trees.

  \myparagraph{The \ptree{} and The \pactree{}.}\yan{also refer to Bo's paper for this paragraph.}
  \hide{
    The \ptree{}~\cite{xxx} is a parallel binary search tree for representing sequences, ordered sets, ordered maps and augmented maps (as defined in~\cite{xxx}).
    The \ptree{} is a purely functional data structure, where each update operation creates a new version of the tree without modifying the original one by using the path copying, so that achieving functional updates in a low cost, e.g., $O(\log n)$ work per point.
    A \ptree{} can be constructed in $O(n\log n)$ work and $O(\log n)$ span.
    For batch updates,
    the key premitive in \ptree{} is the \join{} operation, which takes two trees $L$ with size $n$, $R$ with size $m$ and a key-value pair $(k,v)$, and returns a new tree that contains all elements in $L$, $R$ and $(k,v)$ in $O(\log n-\log m)$ work and span.
    By using the \join{} operation, \ptree{} can efficiently implement various batch operations, such as \union{}, \intersect{} and \diff{} with optimal $O(m\log(\frac{n}{m}+1))$ work and $O(\log n\log m)$ span.
    \myparagraph{The \pactree{}.}
  }
  The \ptree{}~\cite{sun2018pam, blelloch2016just, blelloch2022joinable} can be considered a standard parallel binary search tree (BST) that supports a full interface for construction, single and batch updates, set operations, and many 1D queries.
  The \ptree{} uses either the AVL tree, the red-black tree, the treap, or the weight-balance tree as the underlying structure, and adopts a so-called ``\join{}-based framework'' in a divide-and-conquer manner for achieving high update parallelism.
  Any update in the \ptree{} is immutable, which generates a new version and leads to high space usage.
  The \pactree~\cite{dhulipala2022pac} is built based on the \ptree{} with the leaf wrapping technique.
  Any subtree whose size is no more than a given threshold will be flattened into a compressed leaf node with all points stored in an array.
  In this way, the auxiliary space required for the tree structure becomes negligible compared to the size of the spatial data.
  To support the leaf wrapping technique and achieve both theoretical and practical efficiency, the ``\join{}-based framework'' needs to be redesigned for the \pactree{} to handle the transformation between compressed nodes and normal nodes.

  \begin{figure}
  \centering
  \includegraphics[width=\columnwidth]{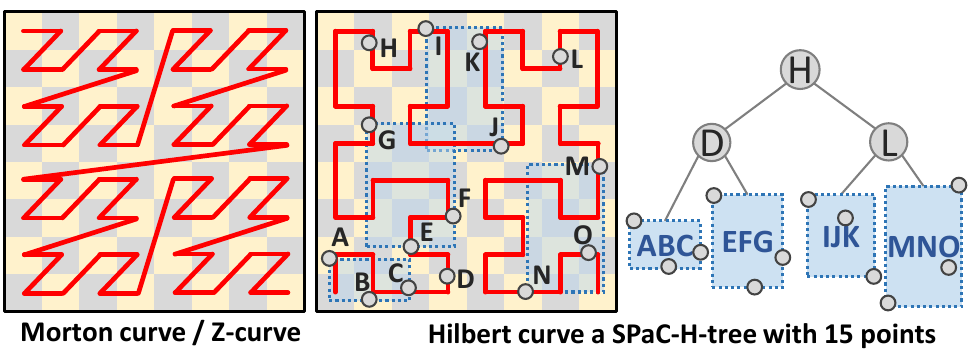}
  \caption{Space-filling curves and an example of a \spactree{} with 15 points and size-3 leaf wrapping.  Each leaf in this case has 3 points and its bounding box marked in blue.}\label{fig:SFC}
\end{figure}

  \myparagraph{Spacial Filling Curves.}
  A \bdita{spatial filling curve (SFC)}
  linearly order multidimensional points such that nearby points in this order are also close in multidimensional space.
  In this case, one can organize the high-dimensional points using data structure for linear sequences.
  A \bdita{\hcurve{}} of order $k$ is a curve $H_k^D$ that visits every vertex of a finite $D$ dimensional grid of size $2^k\times\cdots\times2^k=2^{kD}$ following a specific orientation pattern illustrated in \cref{fig:hilbert},
  which maintains good locality-preserving properties.
  The \bdita{\zcurve{}} is calculated by interleaving the coordinates of a point in binary form to create a single integer, then sorting the points based on their interleaved values as illustrated in \cref{fig:morton}.
  Compared with the \hcurve{}, the \zcurve{} offers more computational efficiency in practice, but preserves less spatial locality for the points.
}

\hide{
  \myparagraph{Geometric Queries.}\yan{move this to problem definition}
  In this paper, we focus on the \bdita{\knn{}} query (find the $k$-th nearest neighbor for a given point), the \bdita{range count query} (count the number of points within a given axis-aligned bounding box), and the \bdita{range report query} (report all points within the given axis-aligned bounding box).

  \myparagraph{Metric Space}
  We denote a metric space as $\mathcal{M}(M,d)$, where $M$ is the collection of points and $d:M\times M\to \R$ is the distance function.
  Let $S\subseteq M$ a finite point set and denote $B_p(r)\subseteq S$ the set of points enclosed by a ball with radius $r$ centered at $p$. Then $S$ has $(\rho,c)-$\bdita{expansion} if and only if $\forall p\in M$ and $r>0$:
  \begin{equation}
    |B_p(r)|\geq \rho \implies |B_p(2r)|\leq c\cdot|B_p(r)|
  \end{equation}
  The constant $c$ is refereed to \bdita{expansion rate} and $\rho$ is usually set to be $O(\log |S|)$. We say the expansion rate is \textit{low} if $c=O(1)$.
  Intuitively, the low expansion property ensures the points distributed uniformly in the space, i.e., it rules out the possibility that a region of dense points locate next to a large patch of empty space.
  Similarly, the \bdita{aspect ratio} $\Delta$ is defined as:
  \begin{equation}
    \Delta = \frac{\max d(x,y)}{\min d(x,y)}\quad \forall x,y\in S
  \end{equation}
  and is said to be \textit{bounded} if $\Delta< n^c$ holds for some constant $c>0$.

  Without loss of generality, we assume the input points have low expansion rate and bounded aspect ratio in the rest of this paper.

  Each tree node is associated with a BVH, which is the tightest bounding volume enclosing all points in its subtree (also known as minimum bounding rectangles).
  The interior ndoes of \rtree{} do not store actual points, and they are used to prune search space for spatial queries.
  All actual points are stored in leaf nodes, whose size are less than given leaf \textit{capacity}.
  \rtree{s} use heuristic strategies to partition the space, and cluster points according to different rules.
  The rules are generally based on spatial properties, such as locality.
  After partitioning the space, each child node will be processed recursively until the number of points to be processed less than given capacity.
}

\section{The Parallel Orth-tree (P-Orth Tree)}\label{sec:orth_tree}
In this section, we introduce our design of the Parallel \orth{} (\porth{}), which partitions points into nested regions recursively based on the spatial median.

\begin{figure}[t]
  \centering
  \includegraphics[width=\columnwidth]{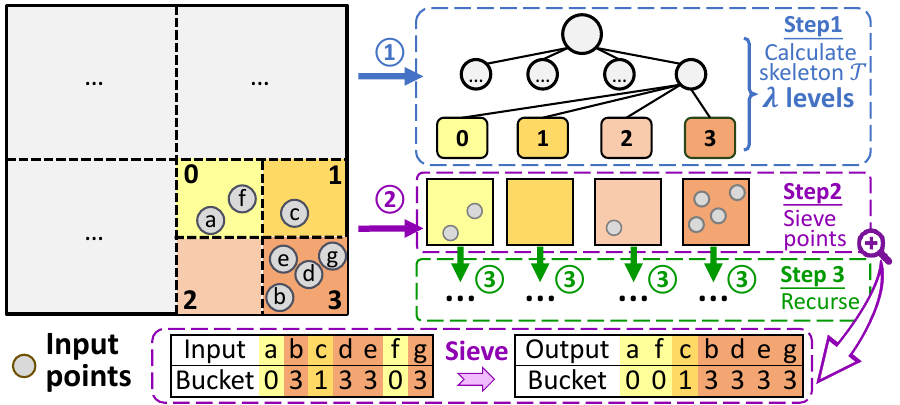}
  \caption{Construction and batch insertion for \porth{s}.}\label{fig:orth-build}
\end{figure}

\myparagraph{Previous algorithms.}
The na\"ive approach to construct or update an \orth{} is to distribute the points to subtrees level by level from the root until reaching the leaves~\cite{finkel1974quad,jackins1980oct}.
However, this approach is slow because the number of rounds of global data movement is proportional to the tree height, which can be large.
Hence, almost all subsequent \orth{s}~\cite{kim2013scalable,warren1993parallel,blelloch2022parallel,pantaleoni2010hlbvh,lauterbach2009fast} use SFCs, specifically the Morton curve (see \cref{fig:SFC}), to speed up the algorithm.
The high-level idea is to sort all input points in Morton order, which only requires $O(\log _M n)$ rounds of global data movement, where $M$ is the cache size.
Then, since \orth{s} always partition at the spatial median, a binary search on the sorted values can identify the partition hyperplane, and all points in one subtree also form a consecutive range in Morton order.
Blelloch and Dobson, in their \zdtree{} paper~\cite{blelloch2022parallel}, also use this idea to achieve a parallel \orthtree{}.

\myparagraph{Issues on Existing Works.}
Although the long-standing Morton-based approach achieves good work, span, and cache bounds, a closer look reveals two major drawbacks. 
\begin{itemize}
  \item \textbf{Performance.} This approach must additionally compute the Morton code for each point as preprocessing and sort the $\langle$code, point$\rangle$ pairs. This increases memory footprint and induces more rounds of reads and writes to all data, which leads to significant overhead (see ``\zdtree{}'' in \cref{fig:heatmap}).
  \item \textbf{Applicability.} While SFCs map higher-dimensional data into one dimension, they suffer from precision limitations. Most modern machines use 64-bit words, which suffices for 2D data (32-bit precision per dimension).
    However, 3D support is constrained to 21 bits per dimension, and handling higher dimensions ($D>3$) is mostly infeasible
    \footnote{It is possible to encode the SFC with higher precision integer, i.e., 128-bit words, but this may also be costly considering bit operations, sorting, and binary-searching in the SFC codes.}.
    Even in lower dimensions, a fallback to the na\"ive partition-based solution is needed when precision is exhausted in certain subregions, which is not elegant. 
\end{itemize}

\setlength{\algomargin}{0em}

\begin{algorithm}[t]
	\fontsize{9pt}{9.5pt}\selectfont
	\caption{Parallel \orth{} (\ourorth) construction\label{algo:orth-constr}}
	\SetKwFor{parForEach}{parallel-foreach}{do}{endfor}
	\KwIn{A sequence of points $P$, region box $H$.}
	\KwOut{A \porth{} $T$ on points in $P$.}
	\SetKw{MIN}{min}
	\SetKw{MAX}{max}
	\SetKwProg{MyFunc}{Function}{}{end}
	\SetKwInOut{Note}{Note}
	\SetKwInOut{Parameter}{Parameter}
	\Parameter{$\skheight$: the height of a tree skeleton.\\
		$\leafwrap$: the leaf wrap of the $k$d-tree.
	}
	\SetKwFor{inParallel}{in parallel:}{}{}

	\DontPrintSemicolon
	\vspace{.5em}
	\MyFunc{\upshape{\orthBuild{$(P, H)$}}}{
		\If{$|P|<\phi$}{
			\Return A leaf node with points $P$ and its bounding box
		}
		Build the tree skeleton $\matht$ by constructing the first $\skheight=(\log_2M)/2D$ levels based on $H$\label{line:buildskeleton}\\
		$B[]\gets $ Split $H$ based on $\matht$\label{line:buildRegion}\tcp*[f]{$B[i]$: the sub-region for bucket $i$}\\
		\tcp*{Reorder points to make those in the same bucket consecutive}  
		$R[]\gets \textsc{Sieve}(P, \matht)$\label{line:obtainSlices}\tcp*[f]{$R[i]$: the slice for all points in bucket $i$}\\
		\parForEach{\upshape external node $i$ of $\matht$}{
			$t\gets \orthBuild(R[i], B[i])$\label{recursive}\tcp*[f]{Recursive build}\\
			Replace the external node $i$ with $t$\\
		}
        Compute the bounding boxes for all internal nodes in $\matht$, and merge non-leaf subtrees with sizes no more than $\phi$\label{line:mergeup}\\
		\Return The root of $\matht$\label{returnroot}
	}

\end{algorithm}

\myparagraph{Our Solution.}
To overcome these issues, our \porth{} design entirely avoids SFCs.
We show that the sorting-based idea can be implemented conceptually equivalently without using SFC.
Consequently, our \ourorth{} is fast and flexible to any coordinate types and ranges (not necessary integers).
Theoretically, the \ourorth{} achieves strong bounds for both construction and batch updates.
Practically, \ourorth{s} outperform \pkdtree{s} and \zdtree{s} in almost all cases, except for very skewed distributions; see \cref{sec:exp} for details.
Below, we present our construction algorithm in \cref{sec:orth_constr}, update algorithm in \cref{sec:orth_update}, and cost analysis in \cref{sec:orth_analysis}.

\begin{algorithm}[t]
	\fontsize{9pt}{9.5pt}\selectfont
	\caption{Batch insertion for \porth{}\label{algo:orth_insert}}
	\SetKwFor{parForEach}{parallel-foreach}{do}{endfor}
	\SetKwProg{MyFunc}{Function}{}{end}
	\SetKwInOut{Note}{Note}
	\SetKwInOut{Parameter}{Parameter}
	\SetKwFor{pardo}{In Parallel:}{}{}
	\KwIn{A sequence of points $P$, a \porth{} $T$ with region $H$.}
	\KwOut{A \porth{} with $P$ inserted.}
	\Parameter{$\skheight$: the maximum height of a fetched tree skeleton.
	}
	\DontPrintSemicolon

	\vspace*{.5em}
	\tcp{The deletion is symmetric.}
	\MyFunc{\upshape{\orthBatchInsert{$(T,P,H)$}}}{
		\lIf{$P=\emptyset$\label{emptyP}}{\Return $T$}
		\If(\tcp*[f]{Insert into a leaf}){\upshape$T$ is a leaf\label{line:base2}}{
			\Return $\orthBuild(T \cup P, H)$\label{rebuildleaf}
		}
		$\matht\gets$ Retrive the skeleton at $T$\label{fetchinsert}\\
		$B[]\gets$ Split $H$ based on $\matht$ \tcp*[f]{$B[i]$: the sub-region for bucket $i$}\label{insertRegion}\\
		$R[]\gets\textsc{Sieve}(P,\matht)$\tcp*[f]{$R[i]$: the slice for points in bucket $i$}\\
		\parForEach{\upshape external node $i$ for $\matht$}{
			$t\gets \orthBatchInsert(i, R[i], B[i])$\label{orthInsertRecursive}\tcp*[f]{Recursive insertion}\\
			Replace the external node of $\matht$ with $t$\\
		}
		Update the bounding boxes of all affected nodes in $\matht$ \label{orthUpdateSkeleton}\\
		\Return The root of $\matht$
	}

\end{algorithm}

\subsection{P-Orth Tree Construction}\label{sec:orth_constr}
Our idea for \porth{} construction is to coordinate the ``conceptual'' sorting process together with the tree construction.
The goal is to build $\skheight$ levels of the tree at once with one round of data movement, and at the same time achieve high parallelism.
Here we adopt the $\textsc{Sieve}(P, \matht)$ function from~\cite{men2025parallel}, which distributes the point set $P$ based on a $\skheight$-level \emph{tree skeleton} $\matht$ in parallel.
At a high level, our algorithm is equivalent to integer sort on Morton codes, on the $\skheight D$ most significant bits in each round.
However, no codes needed to be computed, stored, or compared.
Note that $\textsc{Sieve}()$ is also used in the update algorithms.
We show our \porth{} construction in \cref{algo:orth-constr} and illustrate it in \cref{fig:orth-build}.

\cref{algo:orth-constr} has three steps as shown in \cref{fig:orth-build}.
The first step (lines \ref{line:buildskeleton}--\ref{line:buildRegion}) builds a tree skeleton $\matht$ with $\skheight=(\log_2M)/2D$ levels, where $M$ is the cache size.
This ensures that the number of leaves (external nodes) of $\matht$ is $2^{\skheight\cdot D}$ fits into the cache.
In theory, this step can be done in parallel, although given the small amount of work, in practice this step is run sequentially.
Note that computing the skeleton $\matht$ requires the bounding region for the current subtree, so we also need to compute the corresponding sub-regions for all $\matht$'s leaves (line~\ref{line:buildRegion}).

Once $\matht$ is built, the second step is to sieve the points in $P$ to $\matht$'s leaves.
This step is implemented by the $\textsc{Sieve}(P, \matht)$ function shown on \cref{line:obtainSlices}, introduced by the \pkdtree{} paper~\cite{men2025parallel}. 
This function reorders all points in $P$, such that all points are sorted by the buckets they belong to. The output $R[]$ records the slices of each bucket in the sorted array $P$, such that they can be passed to the recursive calls to deal with each bucket. 
As the number of buckets is usually a small number, the \textsc{Sieve} algorithm resembles a parallel counting sort~\cite{blelloch2020optimal}, which can be performed in an I/O-efficient manner~\cite{blelloch2010low}. 
For page limit, we omit the details of the $\textsc{Sieve}$ algorithm, and directly use the technique as a black box. We refer the readers to~\cite{men2025parallel} for more details. 
An illustration of the result of this step is shown in \cref{fig:orth-build}, and after that,
all points in the same leaf of $\matht$ are gathered together, conceptually stored in an array $R[]$.

\hide{
The method $\textsc{Sieve}()$~\cite{men2025parallel} is to rearrange the input array to make points belonging to the same subtree being contiguous.
This is implemented by dividing $P$ into chunks of size $l$, then counting the number of points in each leaf in $\matht$ in paralle, after which a matrix transpose is performed to compute the offsets for each leaf in each chunk, and finally, all points are distributed to the final destination in parallel.
The final step is to recursively build the \orth{} for each leaf (subtree) in parallel (line~\ref{recursive}).
Once finished, we update the bounding boxes of all internal nodes in $\matht$ (line~\ref{line:mergeup}) and return the root of the skeleton (line~\ref{returnroot}).
}

\subsection{Batch Updates for P-Orth Trees}\label{sec:orth_update}

Both batch insertion and deletion for \porth{s} closely resemble the construction algorithm.
Here we first introduce the batch insertion algorithm, given in \cref{algo:orth_insert}, and discuss the deletion algorithm later.

The batch insertion algorithm takes a batch of points $P$, and adds them to an existing \ourorth{} $T$.
To do so, we sieve the points also for $\skheight$ levels, and then recursively insert points to each bucket in parallel.
One can almost see a one-to-one mapping for these three steps in \cref{fig:orth-build} and \cref{algo:orth_insert}, except for some minor differences in handling base cases.
For deletions, an additional step is needed: for all affected leaves, we flatten their ancestors if the total subtree sizes are smaller than the leaf wrap threshold.
Our update algorithms remain simple since no rebalancing is needed for \orth{s}.


\subsection{Theoretical Analysis}\label{sec:orth_analysis}

Due to page limit, we defer the analysis to the \ifconference{the full version paper~\cite{psiPaperFull}}\iffullversion{\cref{app:porth-analysis}}, and only list the theorems here.

\vspace{-.5em}

\begin{theorem}\label{thm:porth}
  \cref{algo:orth-constr} constructs a \porth{} of size $n$ using $O(n\log \Delta)$ work, $O(\log n\log \Delta)$ span, and $O(n/B\log_M \Delta)$ cache complexity.
  A batch update of size $m=O(n)$ on a \ourorth{} of size $n$ uses $O(m\log \Delta)$ work, $O(\log m\log \Delta)$ span, and $O(m/B\log_M \Delta)$ cache complexity.
\end{theorem}
\vspace{-.5em}
Here, $\Delta$ denotes the aspect ratio, defined as $\frac{\max d(x,y)}{\min d(x,y)}$ for all points $x$ and $y$.
Note that $\log \Delta\ge \log \Theta(n^{1/D})=\Omega(\log n)$ when the point set contains no duplicates in $\mathbb{R}^D$.

With stronger assumptions---for instance, a \emph{bounded aspect ratio} ($\Delta \le n^c$ for some constant $c>0$) and a \emph{constant expansion rate} (full definition in the\ifconference{ the full version paper~\cite{psiPaperFull}}\iffullversion{~\cref{app:porth-analysis}})---we may obtain tighter bounds.
With bounded aspect ratio, we can show that the construction with $O(n\log n)$ work, $O(\log^2 n)$ span, and $O(n/B\log_M n)=O(\mbox{Sort}(n))$ cache complexity.
Updates have $O(m\log n)$ work, $O(\log m\log n)$ span, and $O(m/B\log_M n)$ cache complexity.
With both assumptions, a $k$-NN query can be answered in $O(k\log n)$ work~\cite{blelloch2022parallel}.




\hide{

  However, it requires asymptotically tree height rounds of data movement to finish all partitions, which is less efficient for small cache.
  A widely used alternative pre-sorts the points by their Morton codes and then recursively splits the Morton curve at its midpoint to build a binary tree, known as the \zdtree{}~\cite{kim2013scalable, warren1993parallel,blelloch2022parallel}.
  This achieves good theoretical guarantees by using the integer sort, while it still has drawbacks.
  First, it pipelines the Morton codes generation, sorting, and tree construction, which incurring more rounds of data movement.
  Actually, only the sorting and following tree construction is already slower than the state-of-the-art \kdtree{} (\pkdtree{})~\cite{men2025parallel}, which interleaves the points partition with the node creation.
  This is counter-intuitive as the \orth{} is generally considered to have faster build and update time than the \kdtree{} by its simpler structure, and implies the possible improvements for \orth{} construction.
  Second, reliance on Morton codes limits the \orth{} applicability in higher dimensions ($D>3$) due to the precision issue, which forces a fallback to the naive partition-based method that is less effective due to I/O in-efficiency.
}

\section{The Spatial \titlecap{PaC-tree (\spactree)}}\label{sec:spac}
This section presents the design of the Spatial PaC-tree (\spactree), a highly parallel R-tree with extremely fast construction and updates while maintaining good query speed.

\myparagraph{Existing R-trees.}
As introduced in \cref{sec:prelim}, R-trees are object-partitioning trees, leaving flexibility in the heuristics used to build them.
The original and early designs~\cite{guttman1984r,beckmann1990r,BerchtoldKK96,ohsawa1990new,SellisRF87} are incremental: points are inserted one by one; a greedy strategy iteratively selects a subtree for this point.
When a subtree is much heavier than its siblings, a split is applied by a heuristic (e.g., ``linear''~\cite{guttman1984r,ang1997new}, ``quadratic''~\cite{guttman1984r}, or ``R$^*$''~\cite{guttman1984r, beckmann1990r}).
While simple and highly dynamic, this approach is hard to generalize to parallel batch updates. 
Consequently, prior work on parallel R-trees has primarily focused on parallel queries~\cite{luo2012parallel,you2013parallel,kamel1992parallel,prasad2015gpu} or static construction (bulk loading)~\cite{papadopoulosM03,achakeev2012sort,schnitzerL99master,garcia1998greedy,shuhua2000design,leutenegger1997str,arge2008priority}.
However, for purely static scenarios, \kdtree{s} and \orth{s} are often preferable choices.

A promising approach to parallelize R-trees is via space-filling curves (SFCs).
SFCs map points in higher dimensions to 1D (see \cref{fig:SFC}), enabling all points to be organized in this 1D order using a binary search tree (BST) or a B-tree—equivalently yielding an R-tree if each node maintains its bounding box.
This idea was first noted by Tropf and Herzog~\cite{tropf1981multidimensional}, and later realized in the Hilbert R-tree~\cite{kamel1993packing,haverkort2008four}, which is built atop a B-tree.
Unfortunately, parallel batch update on B-trees can be challenging.
Qi et al.~\cite{qi2020packing} showed that the logarithmic method can sidestep parallel updates for B-trees, but it introduces substantial query overhead~\cite{men2025parallel}.

\myparagraph{The \pactree{}.}
The \pactree{}~\cite{dhulipala2022pac} is a parallel binary search tree (BST) with the leaf-wrapping technique to enable better space- and I/O-efficiency, where
a subtree of size under a threshold $\phi$ (typically 32) is flattened into a compressed leaf stored as an array.
It uses a ``\join{}-based framework'' in a divide-and-conquer manner for high parallelism, and supports the full BST interface, including construction, single and batch updates, and various 1D queries.

\myparagraph{Our SPaC-Tree.}
At first glance, \pactree{s} appear to provide a straightforward solution for parallelizing R-trees:
they can be directly adopted
to support an SFC-based approach, achieving both efficiency and high parallelism.
We implemented this straightforward design and, somewhat unexpectedly, found it much slower than \porth{s} and \pkdtree{s} (see \cpamh{} and \cpamz{} in \cref{fig:heatmap}).
The bottleneck is that a \pactree{} enforces a total order on points according to an SFC, which is overly costly.
In contrast, \porth{s} and \pkdtree{s}
leave points in the leaves unsorted.

To reduce update costs, we introduce the \emp{Spatial-PaC-tree (\spactree{})}.
The primary goal is to keep leaf points unsorted, which requires redesigning and disentangling parts of the underlying \pactree{} algorithms.
At a high level, compared with \pactree{}, the \spactree{} improves over two aspects: 1) integrating the entire construction algorithm into the sorting algorithm by delaying the computation of SFC to the first distribution round of the sorting algorithm, which improves the overall performance, and 2) allowing for unsorted leaves in batch updates, which reduces work for batch update, and has almost no negative impact on queries.

The remainder of this section presents the new design and its analysis.

\setlength{\algomargin}{0em}
\begin{algorithm}[t]
	\fontsize{9pt}{9pt}\selectfont
	\caption{Parallel \ourcpam{} construction\label{algo:ptree-constr}}
	\SetKwFor{parForEach}{parallel-foreach}{do}{endfor}
	\SetKwFor{parFor}{parallel-for}{do}{endfor}
	\KwIn{A sequence of points $P$.}
	\KwOut{A \ourcpam{} $T$ on points in $P$.}
	\SetKw{MIN}{min}
	\SetKw{MAX}{max}
	\SetKwProg{MyFunc}{Function}{}{end}
	\SetKwInOut{Note}{Note}
	\SetKwInOut{Parameter}{Parameter}
	\SetKwFor{inParallel}{in parallel:}{}{}
	\SetKwFor{pardo}{In Parallel:}{}{}

	\DontPrintSemicolon
	\vspace{.5em}
	\MyFunc{\upshape{\cpamBuild{$(P)$}}}{
		$A\gets $ Auxiliary sequence of empty pairs $\ip{code}{id}$ with size $|P|$\\
		$A'\gets$ \sortSfc{$(P, A)$}\label{line:sort}\\
		\Return \cpamBuildFromSorted{$(P, A')$}\label{line:buildCpamFromSorted}
	}

	\vspace{.5em}
	\tcp*{Modify the sample-sort to compute the SFC code with sorting}
	\MyFunc{\upshape{\sortSfc{$(P, A)$}}}{
		Sample points from $P$ and compute their SFC codes\\
		Sort samples and sub-sample them to get the pivots\\
        Partition $P$ into blocks, and compute offsets of blocks as $F[]$\\
		\parFor{\upshape $i$-th block $B$}{
			\parFor{\upshape $j$-th point $p$ in $B$} {
				$k\gets $ The SFC code of $p$\\
				$id\gets $ The id of $p$ \\
				$A[F[i]+j]\gets \ip{k}{id}$\label{line:orth-build-assign}\tcp*[f]{Store the code and id in $A$}\\
			}
			Sort the slice $A[F[i],A[F[i+1])$\label{line:orth-build-sort}\\
			Merge with samples to get counts for each block \label{line:orth-build-merge}
		}
		Redistribute $A$ to buckets $A'$ using the matrix transpose~\cite{axtmann2017place, blelloch2010low}, where the $i$-th bucket has offset $F'[i]$ \label{line:orth-build-redistribute}\\ 
		\parFor(\tcp*[f]{Recursive sorting}){the $i$-th bucket}{
		Sort the slice $A'[F'[i],A'[F'[i+1])$
	}
	\Return The sorted sequence $A'$
	}

	\vspace{.5em}
	\tcp*{Recursively construct the tree.}
	\MyFunc{\upshape{\cpamBuildFromSorted{$(P, A)$}}\label{line:buildFromSorted}}{
	$n\gets |P|$\\
	\If(\tcp*[f]{Input size is below the leaf wrapping}){$n\leq \leafwrap$}{
		Retrieve points $S \subseteq P$ using the ids in $A$ \\
		\Return A leaf node with points $S$ and its bounding box\label{line:base-case}}
	\Else{
	$m\gets n/2$\\
	\pardo{}{
	$L\gets$ \cpamBuildFromSorted{$(P[0, m), A[0, m))$}\label{line:buildLeft}\\
	$R\gets$ \cpamBuildFromSorted{$(P[m+1, n), A[m+1, n))$}\label{line:buildRight}
	}
	}
	$k\gets $ the point in $P$ with id in $A[m]$ \tcp*[f]{The pivot point}\\
	\Return An interior node with left child $L$, right child $R$, pivot $k$, and computing the bounding box from children\label{line:buildNode}
	}
\end{algorithm}

\hide{
	\begin{algorithm}[t]
		\fontsize{8pt}{8.5pt}\selectfont
		\caption{Parallel \ourcpam{} construction\label{algo:ptree-constr}}
		\SetKwFor{parForEach}{parallel-foreach}{do}{endfor}
		\KwIn{A sequence of points $P$.}
		\KwOut{A \ourcpam{} $T$ on points in $P$.}
		\SetKw{MIN}{min}
		\SetKw{MAX}{max}
		\SetKwProg{MyFunc}{Function}{}{end}
		\SetKwInOut{Note}{Note}
		\SetKwInOut{Parameter}{Parameter}
		\SetKwFor{inParallel}{in parallel:}{}{}
		\SetKwFor{pardo}{In Parallel:}{}{}

		\DontPrintSemicolon
		\vspace{.5em}
		\MyFunc{\upshape{\cpamBuild{$(P)$}}}{
			Compute the SFC code for each point in $P$. \label{line:code}\\
			Parallel sort $P$ by the SFC codes.\label{line:sort}\\
			\Return \cpamBuildFromSorted{$(P)$}\label{line:buildCpamFromSorted}
		}

		\vspace{.5em}
		\MyFunc{\upshape{\cpamBuildFromSorted{$(P)$}}\label{line:buildFromSorted}}{
			$n\gets |P|$\\
			\If{$n\leq \leafwrap$}{\Return A leaf node wraps $P$\label{line:base-case}}
			\Else{
				$mid\gets n/2$\\
				\pardo{}{
					$L\gets$ \cpamBuildFromSorted{$(P[0, mid])$}\label{line:buildLeft}\\
					$R\gets$ \cpamBuildFromSorted{$(P[mid+1, n])$}\label{line:buildRight}
				}
				\Return \node($L$, $P[n/2]$, $R$)\label{line:buildNode}\tcp*[f]{To keep the blocked leaves.}
			}
		}

	\end{algorithm}}

\subsection{\titlecap{\spactree{}} Construction}\label{sec:spac_constr}
We first show the construction algorithm for \spactree{s} in \cref{algo:ptree-constr}.
To use \pactree{} for construction, a simple idea is to first compute the SFC code for each point, sort the points accordingly, and then build a balanced BST tree on the sorted points.
Despite theoretical efficiency, directly calling the \pactree{} in \cpam{} in this way is up to 3$\times$ slower than \pkdtree{} construction.
To improve performance, our main effort is to avoid unnecessary memory reads/writes by
redesigning the sorting algorithm, shown in function ``\sortSfc{}'' \cref{algo:ptree-constr}, with two major improvements.
First, instead of pre-calculating SFC values before sorting, we compute them when the points are first touched in sorting, which saves one round of reads and writes to associated arrays.
Second, we only sort the $\ip{code}{id}$ pairs (line~\ref{line:orth-build-assign}), without the coordinates.
This reduces the memory footprint of the recursive sorting process (thus faster speed), at the cost of more cache misses when fetching points to the leaves.
Overall this reduces the running time.
Combining the two techniques together, \cref{algo:orth-constr} can achieve a consistent speedup over the plain implementation (3.1--3.5$\times$ on 2D data; see \cref{fig:heatmap}).

\begin{algorithm}[t]
	\fontsize{8.5pt}{8pt}\selectfont
	\caption{Parallel Batch Insertion on \ourcpam{s}\label{algo:cpaminsert}}
	\SetKwFor{parForEach}{parallel-foreach}{do}{endfor}
	\SetKwProg{MyFunc}{Function}{}{end}
	\SetKwInOut{Note}{Note}
	\SetKwInOut{Parameter}{Parameter}
	\SetKwFor{pardo}{In Parallel:}{}{}
	\KwIn{A sequence of points $P$ and a \ourcpam{} $T$.}
	\KwOut{A \ourcpam{} with $P$ inserted.}
	\DontPrintSemicolon

	\vspace*{.5em}
	\MyFunc{\upshape{\cpamBatchInsert{$(T,P)$}}}{
		Compute SFC codes for points in $P$, and sort $P$ accordingly. \label{line:insert-code}\label{line:insert-sort}\\
	\tcp*{In practice we use the \sortSfc() from \cref{algo:ptree-constr}}

		\Return \cpamBatchInsertFromSorted{$(T, P)$}\label{line:insertCpamFromSorted}
	}

	\vspace{.5em}
	\MyFunc{\upshape{\cpamBatchInsertFromSorted{$(T, P)$}}}{
	$n\gets |P|$\\
	\lIf{$n=0$}{\Return $T$\label{line:insert-base-case}}
	\If{$T$ is a leaf }{
		\If{$|T| + n \leq \leafwrap$}{
			Append $P$ to $T$, and mark $T$ as unsorted\label{line:cpam_insert_append}\\
			Update the bounding box of $T$\\
			\Return $T$\label{line:insert-leaf-case}
		}
		\lElse{
			\Return \cpamBuild($P\cup T$)\label{line:insert-rebuild}
		}
	}
	$k\gets$ the SFC code associated with the pivot in (root of) $T$\label{line:insert-getkey}\\
	$t\gets$ binary search $k$ in $P$ (based on the code)\label{line:insert-binary-search}\\
	\pardo{}{
	$L\gets$ \cpamBatchInsertFromSorted{$(T_\ell, P[0, t))$}\label{line:insert-left}\\
	$R\gets$ \cpamBatchInsertFromSorted{$(T_r, P[t, n))$}\label{line:insert-right}
	}
	Update the bounding box of $T$ based on those of $L$ and $R$\\
	\Return \join($L, T_p, R$)\label{line:insert-join}
	}

	\vspace{.5em}
	\tcp{Return a balanced tree joining $L$ and $R$ with pivot $k$.}
	\MyFunc{\upshape{\join{$(L,k,R)$}}\tcp*[f]{this function remains the same as in \cite{blelloch2016just,dhulipala2022pac}}}{
		\lIf{$L$ is heavier}{\Return \rightJoin($L,k,R$)\label{line:right-join}}
		\lIf{$R$ is heavier}{\Return \leftJoin($L,k,R$)}
		\Return \node{($L$, $k$, $R$)}\label{line:insert-join-node}
	}

	\vspace{.5em}
	\tcp{Recursively check $L$'s right spine until the sub-tree size balances with $R$. Create a new tree node $R'$ with children the two balanced sub-trees, attach $R'$ to $L$, and re-balance $L$.}
	\MyFunc{\upshape{\rightJoin{$(L,k,R)$}}\tcp*[f]{\leftJoin{} is symmetric}}{
		\If(\tcp*[f]{The split terminates here}){$L$ and $R$ is balanced\label{line:right-join-balance}}{\Return \node{($L,k,R$)}\tcp*[f]{Return a balanced tree}}
		$\langle L_\ell, k', L_r\rangle\gets \expose(L)$\tcp*[f]{Expand $L$ into a tree if it is a leaf}\label{line:right-join-expand}\\
		$R'\gets \rightJoin{(L_r,k,R)}$\tcp*[f]{Recursively split the right sub-tree of $L$}\label{line:right_join_get_return}\\
		$L'\gets\node{(L_\ell,k',R')}$\tcp*[f]{Attach the newly balanced tree $R'$ to $L$}\label{line:right_join_attach}\\
		Re-balance $L'$ by rotation\\
		\Return $L'$
	}

	\vspace{.5em}
	\tcp{Expand $T$ into a tree if it is a leaf, and reorder the points if necessary.}
	\MyFunc{\upshape{\expose{$(T)$\label{line:expose}}}}{
		\If{$T$ is a leaf}{
			Re-order the points if $T$ is marked as unsorted\label{line:reorder-in-expose}\\
			Build a perfect balanced tree $T'$ from the sorted points in $T$\\
			\Return $\{T'_\ell, T'_p, T'_r\}$
		}
		\lElse(\tcp*[f]{Return the tree as is}){
			\Return $\{T_\ell, T_p, T_r\}$
		}
	}

	\vspace{.5em}
	\MyFunc(\tcp*[f]{Maintain the leaf wrapping invariant.}){\upshape{\node{$(T_\ell,k,T_r)$}}}{
		Create a node $T$ with pivot $k$, left sub-tree $T_\ell$ and right sub-tree $T_r$.\\
		$n\gets |T|$\\
		\lIf(\tcp*[f]{Leaf wrapping does not apply}){$n>2\leafwrap$}{\Return $T$}
		\ElseIf(\tcp*[f]{Redistribute points in leaves $T_\ell$ and $T_r$}){$n>\leafwrap$\label{line:redistribute}}{
			Sort points in $T_\ell$ and $T_r$ if they are marked as un-sorted.\label{line:sort-in-node}\\
			Redistribute sub-trees of $T$ into two leaf nodes with size $n/2$.\\
			\Return $T$
		}
		\Else(\tcp*[f]{Tree size is below the leaf wrapping, embed it into one leaf}\label{line:flatten-in-node}){
			Flatten $T$ and create a leaf node wrapping it.\\
			\Return this new leaf node
		}
	}
\end{algorithm}

\subsection{Batch Updates on \titlecap{\spactree{s}}}\label{sec:spac_update}

Our \spactree{} builds upon \pactree{}~\cite{dhulipala2022pac}, a parallel BST using the join-based algorithmic framework~\cite{blelloch2016just}.
The high-level idea is to use and only use the \join{} operation for tree rebalancing,
which takes two subtrees $L$, $R$, and a key $k$ in the middle, and returns a new, balanced tree with $L\cup\{k\}\cup R$.
Our key observation here is that, as a spatial index, the order of the points in a leaf, which in this case is based on Hilbert- or Z-Code, does not facilitate spatial queries---queries on a leaf must scan all points anyway.
Therefore, our goal is to carefully redesign the \join-based{} algorithms, such that we can maintain theoretical efficiency, and adapt them best to the spatial index setting by \emph{relaxing the key order in the leaves}.
In our experiments, such an improvement significantly speeds up the updates without sacrificing query performance.

We show the detailed batch insertion algorithm in the \cref{algo:cpaminsert}.
The algorithm begins with computing the SFC code and sorting the inputs.
After sieving points to the leaves,
the algorithm either appends points to the leaf and marks it as unsorted or rebuilds the leaf if its size exceeds the threshold (line~\ref{line:insert-leaf-case} and line~\ref{line:insert-rebuild}).
Next, the standard \join{} operation combines two subtrees $L$ and $R$, and performs the rebalancing (line~\ref{line:insert-join}).
Without loss of generality, we assume $L$ is heavier than $R$,
and the \rightJoin{} operation is called (line~\ref{line:right-join}). The \rightJoin{} recursively splits the right subtree of $L$ until it is possible to return a balanced tree using $R$ (line~\ref{line:right-join-balance}).
When the split reaches a leaf, we expand the leaf into a tree as in \pactree{s} using the \expose{} operation (line~\ref{line:expose}).
The difference is if the leaf is marked as unsorted, we will sort the points first (line~\ref{line:sort-in-node}).
When the split subtree is balanced with $R$ (line~\ref{line:right-join-balance}), we create a new tree node $R'$ with children the two balanced subtrees (line~\ref{line:right_join_get_return}), and attach $R'$ to $L$ (line~\ref{line:right_join_attach}).
Note that the previous leaf expansion may break the leaf wrapping for affected leaves.
In this case, we restore the leaf wrapping by checking the tree size: either directly flatten it into one leaf if the size fits within (line~\ref{line:flatten-in-node}), or redistribute the points into two leaves if necessary (line~\ref{line:redistribute}). We will sort the points first if leaves are marked as unsorted (line:~\ref{line:sort-in-node}).

Despite \cref{algo:cpaminsert} appearing complicated, we can prove its correctness by showing its equivalence to a \pactree{}.
For page limit, we defer the analysis to \iffullversion{\cref{app:spac-analysis}}\ifconference{the full version paper~\cite{psiPaperFull}}.

The batch deletion algorithm is similar to the insertion.
The only difference is that when it reaches a leaf, it removes the points there, marks the leaf as unsorted if necessary, and updates the bounding box.
The invariant of leaf wrapping is maintained the same way as in insertion, i.e., line~\ref{line:insert-join-node} and ~\ref{line:right_join_attach}.

\subsection{Theoretical Analysis}

Due to the page limit, we defer the full analysis to \iffullversion{\cref{app:spac-analysis}}\ifconference{the full version paper~\cite{psiPaperFull}}, and present only the results here.

\vspace{-.5em}

\begin{theorem}{}\label{thm:spac}
For $n$ points with integer coordinates, a \spactree{} with Hilbert- or Z-curve can be constructed in $O(n\log n)$ work, $O(\log n)$ span, and $O(\mbox{Sort}(n))$ cache complexity.  A batch update (insertion or deletion) of size $m$ on a \spactree{} of size $n$ uses $O(m\log n)$ work and $O(\log^2 n)$ span.
\end{theorem}

\vspace{-.5em}

\hide{
\begin{table*}[htbp]
	\small
	\centering
	\setlength\tabcolsep{2.5pt} 
	\renewcommand{\arraystretch}{1.0} 

	%
	\begin{tabular}{cc|cccc|ccccccc|ccccccc}
		\toprule
		\multirow{2}[2]{*}{}                                        & \multirow{2}[2]{*}{\textbf{Solver}} & \multicolumn{4}{c|}{\textbf{Build}} & \multicolumn{7}{c|}{\textbf{Incremental Insert}} & \multicolumn{7}{c}{\textbf{Incremental Delete}}                                                                                                                                                                                                                                                                                              \\
		                                                            &                                     & \textbf{Time}                       & \textbf{IDS}                                     & \textbf{ODS}                                    & \textbf{RR}      & \textbf{10\%}    & \textbf{1\%}     & \textbf{0.1\%}   & \textbf{0.01\%}  & \textbf{IDS}     & \textbf{ODS}     & \textbf{RR}      & \textbf{10\%}    & \textbf{1\%}     & \textbf{0.1\%}   & \textbf{0.01\%}  & \textbf{IDS}     & \textbf{ODS}     & \textbf{RR}      \\
		\midrule
		\multirow{8}[2]{*}{\begin{sideways}Sweepline\end{sideways}} & \ourorth{}                          & 4.63                                & \underline{.959}                                 & \underline{1.64}                                & 1.22             & 5.29             & 5.67             & 5.72             & 9.40             & \underline{.960} & \underline{2.30} & \underline{1.01} & \underline{1.92} & 3.08             & 4.25             & 8.74             & \underline{.963} & 2.27             & \underline{1.02} \\
		                                                            & \pkd{}                              & 5.16                                & .992                                             & 2.09                                            & \underline{1.11} & 13.7             & 24.0             & 35.8             & 39.9             & 1.41             & 8.54             & 1.09             & 10.0             & 24.3             & 36.8             & 62.8             & 1.03             & \underline{1.57} & 1.07             \\
		                                                            & \zd{}                               & 8.30                                & 1.15                                             & 3.53                                            & 1.59             & 5.94             & 4.27             & 6.85             & 13.4             & 1.18             & 3.77             & 1.40             & 5.80             & 4.95             & 12.0             & 27.8             & 1.18             & 3.74             & 1.40             \\
		                                                            & \ourcpamh{}                         & 3.32                                & 2.57                                             & 2.56                                            & 1.19             & 2.98             & 3.09             & 4.12             & 8.85             & 3.67             & 2.83             & 1.29             & 2.32             & 2.50             & 3.23             & \underline{8.37} & 2.74             & 2.72             & 1.26             \\
		                                                            & \ourcpamz{}                         & \underline{3.04}                    & 4.75                                             & 3.64                                            & 1.23             & \underline{2.77} & \underline{2.83} & \underline{3.83} & \underline{8.74} & 4.33             & 4.16             & 1.32             & 2.16             & \underline{2.27} & \underline{3.01} & 9.09             & 4.94             & 4.09             & 1.30             \\
		                                                            & \cpamh{}                            & 10.3                                & 3.05                                             & 2.82                                            & 1.43             & 13.5             & 10.2             & 8.29             & 19.4             & 3.34             & 2.69             & 1.48             & 15.0             & 11.1             & 8.67             & 21.2             & 3.34             & 3.00             & 1.48             \\
		                                                            & \cpamz{}                            & 9.81                                & 5.95                                             & 4.11                                            & 1.47             & 13.1             & 9.96             & 8.04             & 20.4             & 5.64             & 4.47             & 1.53             & 14.4             & 10.7             & 8.52             & 22.7             & 6.36             & 4.50             & 1.54             \\
		                                                            & \boost{}                            & 2812                                & 2.47                                             & 2.11                                            & 4.22             & NA               & NA               & NA               & 1156             & 6.70             & 6.11             & 6.10             & NA               & NA               & NA               & 438              & 2.01             & 1.73             & 4.14             \\
		\midrule
		\multirow{8}[2]{*}{\begin{sideways}Varden\end{sideways}}    & \ourorth{}                          & 12.2                                & .575                                             & \underline{1.05}                                & 1.12             & 13.7             & 11.5             & 12.6             & 26.2             & .568             & \underline{1.20} & \underline{1.01} & 6.79             & 8.48             & 12.5             & 28.3             & .568             & \underline{1.18} & \underline{1.01} \\
		                                                            & \pkd{}                              & 6.10                                & \underline{.524}                                 & 5.50                                            & \underline{1.07} & 12.8             & 25.2             & 32.9             & 53.6             & \underline{.530} & 6.01             & 1.03             & 9.36             & 18.9             & 28.8             & 51.6             & \underline{.536} & 6.50             & 1.04             \\
		                                                            & \zd{}                               & 8.57                                & .825                                             & 6.10                                            & 1.54             & 6.01             & 4.34             & 6.89             & 15.0             & .834             & 6.06             & 1.34             & 5.76             & 5.44             & 16.5             & 41.0             & .840             & 6.16             & 1.36             \\
		                                                            & \ourcpamh{}                         & 3.09                                & 13.3                                             & 3.58                                            & 1.15             & 2.71             & 2.95             & 4.18             & 9.23             & 15.1             & 3.08             & 1.23             & 2.22             & 2.54             & 3.54             & 8.21             & 14.1             & 3.69             & 1.22             \\
		                                                            & \ourcpamz{}                         & \underline{2.94}                    & 23.5                                             & 7.46                                            & 1.19             & \underline{2.52} & \underline{2.68} & \underline{3.77} & \underline{8.74} & 20.9             & 5.30             & 1.27             & \underline{2.03} & \underline{2.27} & \underline{3.26} & \underline{7.80} & 23.9             & 17.1             & 1.26             \\
		                                                            & \cpamh{}                            & 10.0                                & 14.5                                             & 3.86                                            & 1.38             & 13.3             & 10.2             & 8.59             & 19.0             & 15.8             & 3.36             & 1.45             & 14.9             & 11.2             & 9.12             & 19.7             & 15.1             & 3.99             & 1.43             \\
		                                                            & \cpamz{}                            & 9.66                                & 25.5                                             & 8.02                                            & 1.42             & 13.0             & 9.94             & 8.24             & 18.5             & 19.2             & 5.54             & 1.49             & 14.0             & 10.8             & 8.75             & 20.2             & 26.3             & 19.2             & 1.48             \\
		                                                            & \boost{}                            & 870                                 & 2.43                                             & 2.28                                            & 4.21             & NA               & NA               & NA               & 407              & 9.59             & 8.13             & 4.45             & NA               & NA               & NA               & 318              & 2.11             & 1.72             & 4.20             \\
		\midrule
		\multirow{8}[2]{*}{\begin{sideways}Uniform\end{sideways}}   & \ourorth{}                          & 3.23                                & \underline{1.22}                                 & \underline{1.20}                                & 1.14             & 4.27             & 10.3             & 19.7             & 29.7             & \underline{1.27} & \underline{1.25} & \underline{1.04} & \underline{4.32} & \underline{10.2} & \underline{19.6} & \underline{29.9} & 1.51             & 1.49             & 1.18             \\
		                                                            & \pkd{}                              & 3.66                                & 1.30                                             & 1.28                                            & \underline{1.08} & 4.52             & 10.6             & 20.9             & 48.9             & 1.38             & 1.33             & 1.05             & 4.37             & 11.2             & 23.0             & 59.8             & \underline{1.42} & \underline{1.40} & \underline{1.07} \\
		                                                            & \zd{}                               & 7.84                                & 1.45                                             & 1.41                                            & 1.56             & 7.47             & 11.5             & 21.0             & 39.4             & 1.55             & 1.53             & 1.44             & 8.12             & 13.2             & 23.4             & 39.8             & 1.49             & 1.47             & 1.39             \\
		                                                            & \ourcpamh{}                         & 3.34                                & 4.66                                             & 4.63                                            & 1.17             & 3.90             & 6.00             & 13.0             & 26.8             & 4.66             & 4.66             & 1.25             & 4.69             & 14.4             & 27.7             & 42.2             & 4.63             & 4.64             & 1.24             \\
		                                                            & \ourcpamz{}                         & \underline{3.10}                    & 22.1                                             & 22.2                                            & 1.22             & \underline{3.68} & \underline{5.75} & \underline{12.6} & \underline{26.2} & 23.1             & 23.2             & 1.31             & 4.78             & 14.2             & 27.4             & 42.0             & 22.6             & 22.6             & 1.30             \\
		                                                            & \cpamh{}                            & 11.3                                & 5.91                                             & 5.87                                            & 1.39             & 15.3             & 38.2             & 108              & 159              & 5.85             & 5.82             & 1.46             & 17.2             & 39.0             & 107              & 157              & 5.84             & 5.81             & 1.46             \\
		                                                            & \cpamz{}                            & 10.8                                & 27.7                                             & 27.8                                            & 1.45             & 14.7             & 37.6             & 108              & 158              & 28.6             & 28.7             & 1.54             & 16.7             & 38.4             & 106              & 157              & 27.8             & 27.9             & 1.52             \\
		                                                            & \boost{}                            & 2269                                & 14.2                                             & 11.6                                            & 5.48             & NA               & NA               & NA               & 908              & 26.1             & 25.2             & 7.81             & NA               & NA               & NA               & 961              & 4.93             & 4.46             & 4.60             \\
		\bottomrule
	\end{tabular}%

	\vspace{.2em}
	\caption{\textbf{Running time (in seconds) for and other baselines. Lower is better.} \ziyang{todo}
		\ziyang{add special mark for boost in updates}
		\vspace{-0.5em}
	}
	\label{table:summary}
\end{table*}}

\begin{figure*}
  \includegraphics[width=\textwidth]{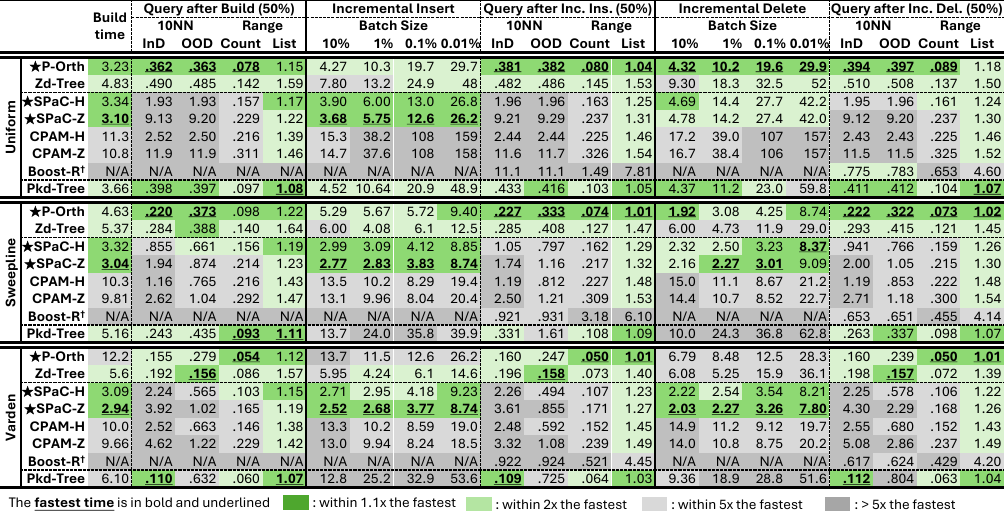}
  \caption{{\bf Running time (in seconds) on synthetic data. Lower is better.} \normalfont 
  The fastest time in each test is in bold and underlined. 
  We use colors to mark results within 1.1$\times$, 2$\times$, 5$\times$, and $>5\times$ the fastest time. 
  Detailed settings for build, queries, and incremental insertion/deletion are introduced at the beginning of \cref{sec:exp-synthetic}.
  \ids{}/\ods: in-/out-of-distribution.   
  $\dagger$: Boost R-tree is sequential and only support point updates. 
  Therefore, we omit the construction/update times, and report query times after incremental inserting/deleting points one by one. 
  }\label{fig:heatmap}
  \vspace{-0.2em}
\end{figure*}

\section{Experiments\label{sec:exp}}
We conduct in-depth experiments to understand the performance of \ourlib{} and other spatial indexes on both synthetic and real-world datasets.
We show that both \ourorth{s} and \ourcpam{s} achieve superior construction and update performance,
outperforming \pkdtree{} in most cases, and are much faster than existing \orth{} and R-tree baselines.
Both \ourorth{s} and \ourcpam{s} also exhibit comparable or better query performance to their corresponding counterparts in prior work.
In addition to showing the effectiveness of our new algorithms, we believe our experiments also provide the first systematic study of various
parallel spatial indexes, including \kdtree{s}, \orth{s}, and R-trees.

\myparagraph{Setup.}
We use a machine with 112 cores (224 hyperthreads) with four Intel Xeon Platinum 8176 CPUs and 1.47 TB RAM.
\ourlib{} is in C++ and compiled using GCC 14.2.1 with \texttt{-O3}.
We use the ParLaylib~\cite{blelloch2020parlaylib} for fork-join parallelism.
Our code is available at~\cite{psiCode}.
We report numbers as the average of 3 runs after a warm-up run. 
More details about parameter choosing are shown in \iffullversion{\cref{app:impl}}\ifconference{the full version paper~\cite{psiPaperFull}}.

\myparagraph{Baselines.} We compare to the following baselines.
\begin{itemize}
  \item \textbf{\pkd{s}}~\cite{men2025parallel}: The state-of-the-art parallel \kdtree{}.
  \item \textbf{\zdtree{s}}~\cite{blelloch2022parallel}: The state-of-the-art parallel \orth{}.
    \zdtree{} uses Morton code to presort the data to aid the construction and update algorithm in a standard \orth{}.
    The original code from~\cite{blelloch2022parallel} has known bugs in the update algorithms (confirmed by the authors).
    We use our own implementation based on their paper.  We have carefully verified that our construction time is similar to their code. 

  \item \textbf{\cpam{}}~\cite{dhulipala2022pac}: 
    As a baseline, we use \pactree{s} from the CPAM library (as a black box)
    to store each point's SFC code as the key. It preserves a total order of
    all points based on the Morton curve (\cpamz{}) or the Hilbert curve
    (\cpamh{}). This baseline highlights how our new design by maintaining
    only a partial order improves performance.

  \item \textbf{\boost{} R-trees}~\cite{schaling2011boost}: The \rtree{} from the Boost library.
    Boost R-tree is sequential, and only supports point updates (no batch updates).
    We mainly use it as a baseline to verify the query performance for our \ourcpam{s}.
    Hence, among all the variants, we use the \texttt{quadratic} version, which gives the best tree quality in the dynamic setting.
\end{itemize}

Within \ourlib{}, we tested the parallel \orth{}---the \porth{}---as introduced in \cref{sec:orth_tree},
and two R-trees---\hilberttree{} and \mortontree{}---which use Hilbert and Morton curve, respectively, on the \cpamtree{} detailed in \cref{sec:spac}.
We maintain bounding boxes for all tested indexes. We refer readers to \ifconference{the full version paper~\cite{psiPaperFull}}\iffullversion{the~\cref{app:impl}} for more implementation details.




\subsection{Overall Evaluation under Synthetic Datasets\label{sec:exp-synthetic}}
\myparagraph{Setup.}
We test different distributions for points, queries, and update patterns of synthetic data.
All coordinates are 64-bit integers in $[0, 10^9]$.
We use three workloads: \uniform{}, \sweep{} and \varden{}.
\uniform{} draws each point uniformly random from the space.
\sweep{} also uses uniform data, but sorts all points along the first dimension.
This is used to simulate a skewed update pattern, where the updated points exhibit spatial locality.
\varden{}~\cite{gan2017hardness} is generated by randomly walking in the space with a low probability to restart at a random position.
Points are clustered and different clusters are far from each other, simulating a skewed point distribution.

We test both static and dynamic cases.
Besides directly measuring the tree construction time, 
we also use the \emph{incremental insertion/deletion} workload with various batch sizes to simulate a highly dynamic scenario.
For batch size $b$, an incremental insertion workload means to construct the index by $n/b$ batch insertions progressively,
and vice versa for deletions (deleting the index in $n/b$ batches).
We report the total running time of all operations.
This reflects how the \emph{update efficiency} of each index is affected under a constantly evolving dataset.
Under this workload, we further time the queries after half of the batches.
The query performance reflects how the \emph{quality} of each index is affected after massive updates.
For the static setting, we also provide query times after building a tree with half of the data for easy comparison with the dynamic setting.
We also test the update time for a single batch, and show the results in \ifconference{the full version paper~\cite{psiPaperFull}}\iffullversion{the~\cref{app:batch-updates}}.


We tested \knn{} and range queries (introduced in \cref{sec:prelim}).
We run $10^7$ 10-NN queries for both in-distribution (\ids) and out-of-distribution (\ods) queries.
For range queries, we test $5\times 10^4$ range-count and range-list queries, with range sizes $10^4$--$10^6$.
Different queries run in parallel. 
Besides \cref{fig:heatmap}, we further study how \knn{} and range-list performance changes with their output sizes
and show results in \cref{fig:knn,fig:rr}.

We summarize in \cref{fig:heatmap} the performance of all tested indexes on synthetic datasets with $10^9$ 2D points.
We provide the results on 3D points in \iffullversion{\cref{app:3d-syn}}\ifconference{the full version paper~\cite{psiPaperFull}}.
Next we analyze the performance in detail.

\subsubsection{Construction}\label{sec:exp:construction}
For tree construction, our \cpamtree{} is the fastest among all indexes across all workloads.
The advantage comes from embedding 2D data into 1D that simplifies the computation,
and various optimizations in \ourlib{} introduced in \cref{sec:spac}.
\mortontree{} is slightly faster than \hilberttree{}, since Morton code has simpler computation than Hilbert code.
The baselines \cpamh{} and \cpamz{} are about $3\times$ slower than our \cpamtree{s}, due to the overhead in maintaining the ordering in leaves.
This effect is even more significant in batch updates and queries.
This justifies the necessity of our technique of relaxing the ordering in leaves.

\hide{
  For \orth{s}, on \uniform{} and \sweep{}, \porth{} also achieves good performance (within 52\% slower than the fastest \mortontree{}),
  and is faster than all other baselines.
  On \varden{}, \porth{} becomes slower than others.
  Since the \orth{} split the space using the coordinate median, it is naturally not resistant to skewed data,
  and is most affected by the skewed distribution among all indexes.
  Although \zdtree{} is also a \orth{}, it achieves reasonable performance---
  the main cost for \zdtree{} construction is to sort all points in Morton order,
  and this is done by a comparison sort in our implementation.
  All other indexes are comparison-based, and the effect of skewed data to them is minimal in construction time.
  In summary, \porth{} has the best or close to the best construction performance in non-skewed data, but exhibits a disadvantage on skewed data.

  Compared to \porth{}, \pkdtree{} has similar but usually slower performance (except for the skewed \varden{} data).
  This is because on uniform data, the object median and coordinate median are close to each other, and thus the chosen splitters of \porth{} and \pkdtree{s} are similar.
  In this case, \porth{}'s \quadtree{} structure allows for shallower tree height and better locality than the binary \kdtree{}.
  In addition, determining the splitter at each tree node in a \porth{} (computing the middle of the coordinate range)
  is also simpler than \kdtree{} (estimating the median among all points).
  As a result, \pkdtree{} is slower than \porth{} on uniform data.
  In all cases, \pkdtree{} is much slower than our \cpamtree{s} (1.18--2.07x slower) in construction.
}

For \orth{s}, on \uniform{} and \sweep{}, the \porth{} also achieves good performance (within 52\% slower than the fastest \mortontree{}),
and is faster than all other baselines.
The advantage of \porth{s} over \pkdtree{s} come from two aspects:
1) as a \quadtree{}, \porth{} allows for shallower tree height and better locality than the binary \kdtree{},
and 2) determining the splitter at each node in a \porth{} (computing the middle of the coordinate range)
is simpler than \kdtree{} (estimating the median among all points).

On \varden{}, \porth{} becomes slower than others.
Since \orth{s} split the space using the coordinate median, it is naturally not resistant to skewed data,
and is most affected by the skewed distribution.
Although \zdtree{} is also a \orth{}, it achieves reasonable performance---
the main cost for \zdtree{} construction is to sort all points in Morton order,
and this is done by a comparison sort in our implementation.
All other indexes are comparison-based, and the effect of skewed data on them is minimal in construction time.

In summary, \cpamtree{s} have consistently better performance than all other baselines in construction.
\porth{} is also competitive on non-skewed data, but exhibits a disadvantage on skewed data.

\subsubsection{Incremental Batch Updates}
The conclusions for batch updates are very similar to those of construction. 
\cpamtree{s} has the best overall performance, and \mortontree{} has a slight advantage over \hilberttree{}.
\mortontree{} is the fastest in all incremental insertions, and most cases in incremental deletions.
For the same reason analyzed in \cref{sec:exp:construction}, \ourorth{s} are less ideal for \varden{} data.
In all other cases, \ourorth{s} are either the best or close to the best.

For all indexes, the incremental update time increases when the batch size decreases.
On the one hand, smaller batches result in less potential for parallelism.
On the other hand, having more batches also means more modifications to the tree, requiring more effort to rebalance the tree and leaving the tree further from being perfectly balanced.
The only index that avoids rebalancing is the \orth{}, and its performance with continuous updates is the least affected by the batch size.

On highly dynamic data, \pkdtree{s} are less competitive in update time compared to \porth{s} and \cpamtree{s}.
One essential reason is that \pkdtree{} has $O(\log^2 n)$ amortized cost per updated point, while \porth{} and \cpamtree{} have cost of $O(\log \Delta)$ and $O(\log n)$, respectively, where $\Delta$ is the aspect ratio.
Hence, both \porth{s} and \cpamtree{s} are faster than \pkdtree{s} in updates.
In particular, \porth{s} are up to 7.18$\times$ faster than \pkdtree{} in incremental updates, and \cpamtree{} can be up to 7.5$\times$ faster.
Even for \varden{} where $\Delta$ is relatively large, \porth{s} are almost always faster than \pkdtree{s} in incremental updates.


\subsubsection{Queries}
We run queries in three settings: 1) after constructing a tree of size $5\times 10^8$, 2) after applying 50\% of the insertion batches, and 3) after applying 50\% of the deletion batches.
Most indexes are nearly perfectly balanced after construction, and thus the first setting reflects their best-case (static) query performance.
The other two settings reflect how the index quality is affected by updates.
In \cref{fig:heatmap}, we only select results for 10-NN query and a relatively large range query.
To give more details, in \cref{fig:knn} and \cref{fig:rr}, we further show how query performance changes with the output size, i.e., $k$ in \knn{},
and the range size in range-list queries.

\begin{figure}[t]
	\includegraphics[width=0.48\textwidth]{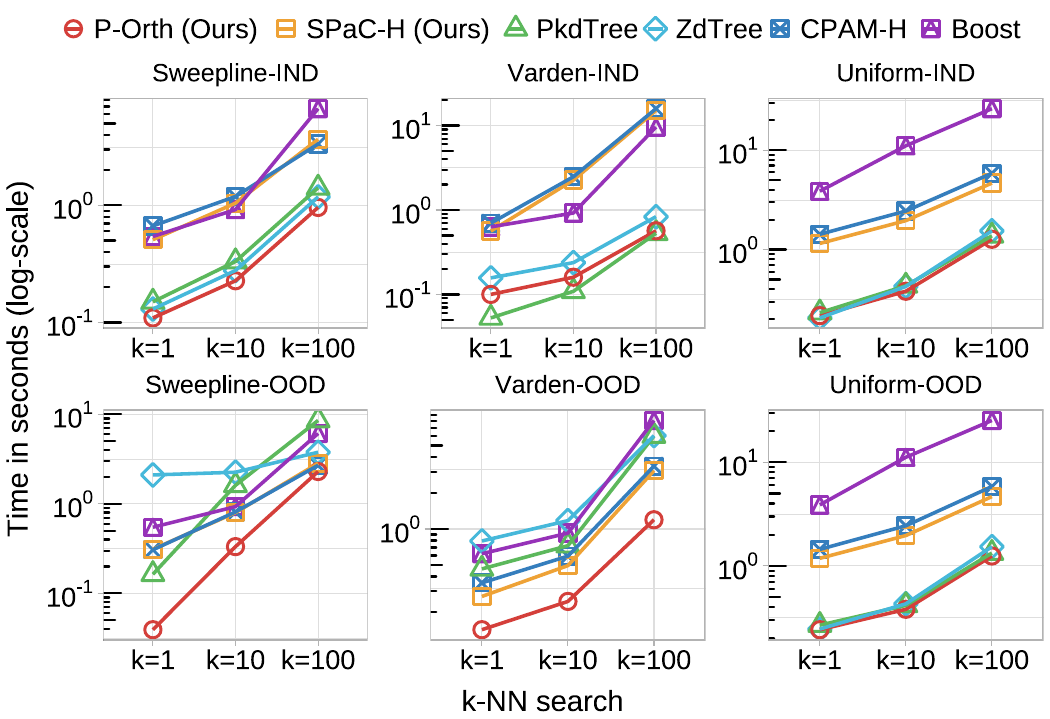}
	\caption{\textbf{Running time (in seconds) of $k$-NN queries for $k\in\{1,10,100\}$. Lower is better.}
		The dataset contains 500M points in 2 dimensions.
		The tree is constructed by incremental insertion with batch ratio $0.01\%$.
		The test contains $k$-NN queries from $10^7$ points from both \ids{} and \ods{} distribution.
		Plots are in log-log scale.
	}
	\label{fig:knn}
	\includegraphics[width=0.48\textwidth]{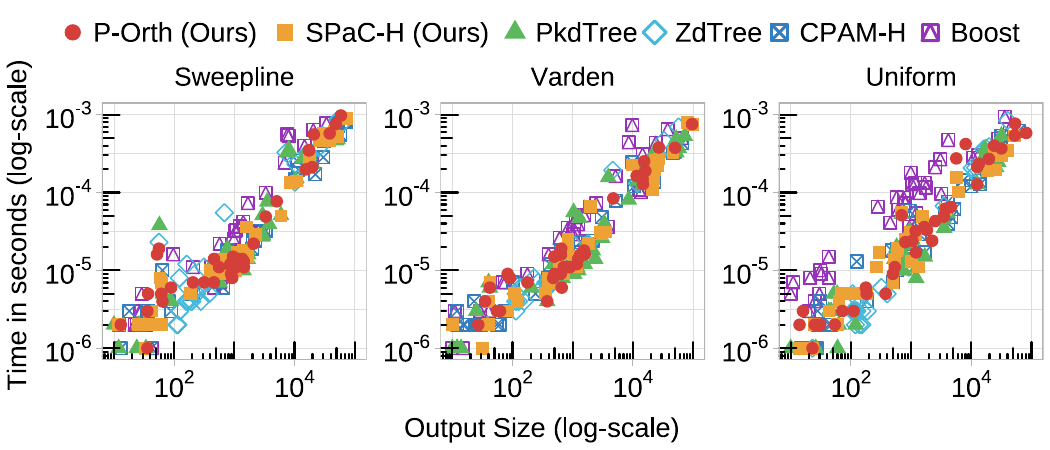}
	\caption{\textbf{Running time (in seconds) of range report queries for w.r.t output sizes. Lower is better.}
		The dataset contains 500M points in 2 dimensions.
		The tree is constructed by incremental insertion with batch ratio $0.01\%$.
		Plots are in log-log scale.
	}
	\label{fig:rr}
\end{figure}

\myparagraph{\knn{} Queries.}
As shown in \cref{fig:knn}, space-partitioning trees are evidently faster than R-trees in \knn{} queries.
This is natural due to overlapping bounding boxes in R-trees.
For \ourcpam{s}, while the \hilberttree{} is slightly slower than \mortontree{} in construction and updates, it is much more efficient in queries.
This is because the Hilbert curve has better locality than the Morton curve (adjacent codes are always geometrically close to each other).
Among the R-trees, \ourcpam{s} achieve similar or better performance than Boost R-tree---in all queries,
\hilberttree{} is between 3.7$\times$ slower to 5.66$\times$ faster,
with a geometric mean of 2.5$\times$ faster.

Among the space-partitioning trees, \orthtree{s} has the best overall performance.
This is because when visiting a subtree, the \ourorth{} can select 1 out of 4 quadrants,
which is more effective than \pkd{s} and \zd{s} that select 1 out of 2 half spaces.
Hence, \pkdtree{s} and \zdtree{s} are competitive but usually slower than \porth{s}.
The only exception is on \varden{} data.
For \ids{} queries, due to the skewed distribution of \varden{}, \orthtree{s} may be unbalanced,
and thus the comparison-based \pkdtree{s} perform better.
Interestingly, on the contrary, both \orthtree{s} exhibit an advantage on \ods{} queries on \varden{}.
The reason is still in imbalance---for \varden{}, points are highly clustered, making these regions in the tree deep and other regions shallow.
Since the \ods{} queries distribute differently from the input, they likely hit the shallow regions and thus are much faster.

\myparagraph{Range Queries.} As shown in \cref{fig:heatmap,fig:rr},
\pkdtree{s} show a small but consistent advantage on range queries.
This is because a range query visits all subtrees overlapping the query box.
In this case, \ourorth{s} have to explicitly check the bounding boxes for four subtrees,
while every non-overlapping check on a \pkd{} node can prune half of the points in this subtree.
For other indexes, the relative performance on range queries is similar to \knn{} queries.
Interestingly, while \ourcpam{s} are still slower than \kdtree{s} and \porth{s} in range-list queries, the difference is much smaller, especially on large ranges---
in this case, the query time is mostly spent emitting the result list, hiding the difference in pruning effectiveness across indexes.
Therefore, range queries are less sensitive to the index type than \knn{} queries.

\myparagraph{Impact of Updates to Queries.}
In the dynamic setting, the \orthtree{s} (\porth{} and \zdtree{}) are history-independent (modulo leaf-wrapping), namely,
the final state of the tree is not affected by the operation order.
Therefore, their query performance is least affected by batch updates, and is the best in the dynamic setting.

For all other indexes, the tree may get less balanced after updates.
Indeed, they all get slower to some extend compared to the static setting.
This impact is moderate for most indexes (mostly within 20\%).
The exceptions all appear in \ods{} \knn{} queries, where
\pkdtree{} gets 3.7$\times$ slower after incremental insertion on \sweep{},
and \cpamz{} and \mortontree{} get about 2.5$\times$ slower after incremental deletion on \varden{}.


In summary, for queries,
\orth{s} and \kdtree{s} are naturally better than R-trees.
\kdtree{s} are better in dealing with \ids{} queries on non-uniform data, but may be worse in \ods{} queries.
\ourorth{} has the best or close to the best query performance in almost all queries and workloads.

\subsection{Operations on Real-World Datasets\label{sec:exp-real-world}}

\hide{
\begin{table*}[t]
	\centering

	\small
	\setlength\tabcolsep{2pt}
	\renewcommand{\arraystretch}{0.8}

	\begin{tabular}{cc|ccccc|ccccc|ccccc}
		\toprule
		                                                          &                 & \multicolumn{5}{c|}{\textbf{Build}} & \multicolumn{5}{c|}{\textbf{Insert}} & \multicolumn{5}{c}{\textbf{Delete}}                                                                                                                                                                                       \\
		                                                          & \textbf{Solver} & \textbf{Time}                       & \textbf{IDS}                         & \textbf{ODS}                        & \textbf{RCL} & \textbf{RRL} & \textbf{Time} & \textbf{IDS} & \textbf{ODS} & \textbf{RCL} & \textbf{RRL} & \textbf{Time} & \textbf{IDS} & \textbf{ODS} & \textbf{RCL} & \textbf{RRL} \\
		\midrule
		\multirow{8}[2]{*}{\begin{sideways}GeoLife\end{sideways}} & OrthTree        & .0064                               & .0007                                & .0010                               & .057         & .447         & .079          & .0006        & .0009        & .047         & .432         & .089          & .0009        & .0009        & .056         & .436         \\
		                                                          & KdTree          & .0053                               & .0003                                & .0003                               & .056         & .431         & .131          & .0002        & .0003        & .110         & .436         & .147          & .0003        & .0003        & .058         & .433         \\
		                                                          & ZdTree          & -                                   & -                                    & -                                   & -            & -            & -             & -            & -            & -            & -            & -             & -            & -            & -            & -            \\
		                                                          & PTree-H         & .0015                               & .0010                                & .0009                               & .103         & .537         & .240          & .0010        & .0008        & .108         & .540         & .137          & .0013        & .0009        & .110         & .541         \\
		                                                          & PTree-Z         & .0012                               & .0008                                & .0011                               & .122         & .539         & .205          & .0010        & .0013        & .130         & .543         & .126          & .0009        & .0012        & .131         & .541         \\
		                                                          & CPAM-H          & .0036                               & .0012                                & .0009                               & .114         & .544         & .480          & .0011        & .0009        & .122         & .538         & .481          & .0011        & .0010        & .122         & .541         \\
		                                                          & CPAM-Z          & .0032                               & .0010                                & .0012                               & .137         & .542         & .449          & .0010        & .0013        & .147         & .542         & .453          & .0011        & .0014        & .147         & .547         \\
		                                                          & Boost           & .213                                & .0008                                & .0008                               & .116         & 7.83         & .087          & .0006        & .0004        & .084         & 8.05         & .064          & .0005        & .0003        & .116         & 7.77         \\
		\midrule
		\multirow{8}[2]{*}{\begin{sideways}Cosmo50\end{sideways}} & OrthTree        & 1.90                                & .417                                 & .554                                & .566         & 1.74         & 16.2          & .431         & .582         & .568         & 1.66         & 17.9          & .463         & .617         & .616         & 1.75         \\
		                                                          & KdTree          & 1.89                                & .364                                 & .378                                & .602         & 1.70         & 101           & .388         & .373         & .626         & 1.67         & 800           & .402         & .539         & .620         & 1.67         \\
		                                                          & ZdTree          & -                                   & -                                    & -                                   & -            & -            & -             & -            & -            & -            & -            & -             & -            & -            & -            & -            \\
		                                                          & PTree-H         & 1.02                                & 1.28                                 & 1.89                                & .764         & 1.76         & 6.59          & 1.74         & 2.22         & .840         & 1.95         & 9.40          & 1.64         & 1.91         & .837         & 1.94         \\
		                                                          & PTree-Z         & .837                                & 8.90                                 & 11.9                                & .980         & 1.93         & 5.63          & 9.38         & 11.7         & 1.06         & 2.12         & 8.46          & 9.37         & 13.5         & 1.06         & 2.12         \\
		                                                          & CPAM-H          & 5.48                                & 1.52                                 & 2.56                                & 1.04         & 2.13         & 19.6          & 1.86         & 2.71         & 1.06         & 2.22         & 19.8          & 1.93         & 2.44         & 1.06         & 2.23         \\
		                                                          & CPAM-Z          & 5.38                                & 14.5                                 & 17.6                                & 1.30         & 2.33         & 19.0          & 13.0         & 14.0         & 1.33         & 2.44         & 19.2          & 13.1         & 18.3         & 1.33         & 2.44         \\
		                                                          & Boost           & 414                                 & 4.77                                 & 4.04                                & .977         & 19.3         & 193           & 6.45         & 6.42         & 1.58         & 19.9         & 1483          & 10.5         & 9.05         & 1.21         & 15.0         \\
		\midrule
		\multirow{8}[2]{*}{\begin{sideways}OSM\end{sideways}}     & OrthTree        & 4.96                                & .388                                 & .442                                & .050         & 1.13         & 14.5          & .381         & .439         & .045         & 1.04         & 14.9          & .384         & .452         & .046         & 1.06         \\
		                                                          & KdTree          & 4.32                                & .378                                 & .499                                & .049         & 1.06         & 29.3          & .387         & .566         & .054         & 1.01         & 26.3          & .384         & .520         & .051         & 1.01         \\
		                                                          & ZdTree          & 5.88                                & .572                                 & .880                                & .055         & 1.57         & 16.5          & .566         & .880         & .058         & 1.36         & 23.9          & .572         & .879         & .055         & 1.35         \\
		                                                          & PTree-H         & 2.26                                & 6.13                                 & 6.56                                & .085         & 1.17         & 8.19          & 6.12         & 5.85         & .096         & 1.22         & 7.98          & 5.19         & 6.30         & .088         & 1.24         \\
		                                                          & PTree-Z         & 2.12                                & 17.6                                 & 15.7                                & .131         & 1.19         & 7.91          & 20.2         & 16.5         & .152         & 1.27         & 7.37          & 17.8         & 16.3         & .138         & 1.27         \\
		                                                          & CPAM-H          & 7.26                                & 6.61                                 & 7.11                                & .117         & 1.39         & 15.6          & 5.43         & 5.74         & .129         & 1.42         & 16.7          & 5.50         & 6.59         & .125         & 1.44         \\
		                                                          & CPAM-Z          & 7.01                                & 19.9                                 & 17.1                                & .182         & 1.43         & 15.4          & 20.4         & 16.0         & .199         & 1.47         & 16.6          & 20.0         & 17.9         & .194         & 1.48         \\
		                                                          & Boost           & 666                                 & 9.95                                 & 8.95                                & .435         & 8.35         & 312           & 20.4         & 19.3         & .562         & 10.1         & 264           & 4.69         & 4.06         & .516         & 7.50         \\
		\bottomrule
	\end{tabular}%

	\caption{\textbf{
			Tree construction and $k$-NN time on read-world datasets for and baselines. Lower is better.} The ``Points'' is the number of points in the datasets and ``Dim.'' is the dimension for the points. \knn{} queries are performed in parallel on all points in the dataset. 
		``Range'' is the time for $10^3$ range report queries with output size between $10^4$--$10^6$.
		The fastest runtime for each benchmark is underlined.
		``s.f.'': segmentation fault. ``t.o.'': time out (more than 3 hours).
	}
	\label{table:real-world}%
\end{table*}%
}

\begin{figure}
\includegraphics[width=\columnwidth]{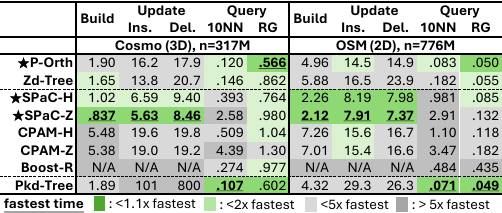}
\caption{\textbf{Running time (in seconds) on real-world datasets. Lower is better.} Insert/Delete: incremental insertion/deletion/ with batch size 0.01\%. 
``RG'': Range-list queries. 
\label{fig:real-heatmap}
}
\end{figure} 
\begin{figure}[t]
	\includegraphics[width=0.48\textwidth]{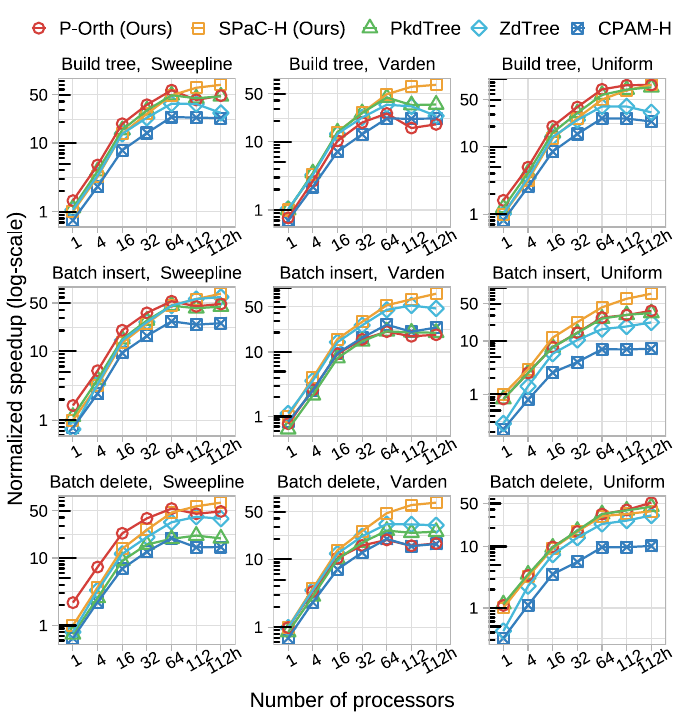}
	\caption{\textbf{Normalized parallel speedup of tree construction, batch insertion and batch deletion on different number of processors. Higher is better.}
		The speedup is normalized to the running time on the \spactree{} on 1 thread respectively.
		The dataset contains 1000M points in 2 dimensions. "Batch insert" is to insert another 1\% points into a tree containing 1000M points, and "Batch delete" is to delete 1\% points from the tree. Both in a single batch.
		"112h" means using all 112 cores (224 hyper-threads).
		There is no result for \boost{} since it is sequential.
	}
	\label{fig:scalability}
\end{figure}

\begin{figure*}
  \includegraphics[width=0.9\textwidth]{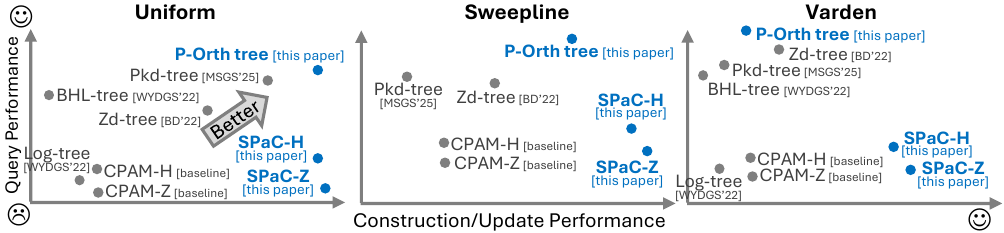}
  \caption{\textbf{Summary of tested index in update and query performance.}
    Results are summarized from numbers in \cref{fig:heatmap}. In particular, the data points are based on the geometric mean of all relevant operations (updates or queries) in \cref{fig:heatmap}.
    Data points for \logtree{} and \bhltree{}~\cite{yesantharao2021parallel} are estimated from the \pkdtree{} paper~\cite{men2025parallel}. Our new algorithms are marked in blue.
  We note that this figure only gives the average of the tested benchmarks in this paper. More comprehensive conclusions can be found in \cref{table:related-work}.\label{fig:summary}}
\end{figure*}

For real-world datasets, we test a highly clustered dataset \texttt{COSMO}~\cite{scoville2007cosmic} and the \texttt{OpenStreetMap} (\osm{}~\cite{haklay2008openstreetmap}) for Northern America.
We test $10^7$ \ids{} 10-NN quiries, and $5\times 10^4$ range-list queries with range size $10^4$--$10^6$.
Coordinates are rounded down to 64-bit integers. We remove duplicates and shift all points to positive coordinates.
To ensure the SFC works properly in 3D, we scale the coordinates to $[0, 10^6]$.
We evaluate construction, incremental updates with batch ratio 0.01\% and queries after construction, and show results in \cref{fig:real-heatmap}.

\cpamtree{s} are much faster than others in construction and updates.
On real-world data, this advantage is more significant than synthetic data.
In particular, they are about 2$\times$ faster than \pkdtree{s} in construction, and 3.5--94$\times$ faster in updates.
\porth{s} have similar construction times to \pkdtree{s}, but are much faster in updates (1.8--44.7$\times$ faster).

On queries, R-trees still perform worse than space-partitioning trees.
In most of the cases, \pkdtree{} achieves the best query performance, but \porth{s} are always competitive---in all cases, the difference in within 20\%. Considering that \porth{s} are 1.8--44.7$\times$ faster in updates, \porth{s} offer a much better query/update tradeoff than \pkdtree{s}.

\hide{\myparagraph{Tree construction}
We measure the construction time on the entire dataset, and then perform queries on a tree directly built with half of the data.
\mortontree{s} and \hilberttree{s} are faster than others in construction for all benchmarks, and the \ourorth{} is slightly slower than \pkdtree{s}, up to x$\times$, but still achieve xx--xx$\times$ speedup than the \zdtree{}.
\pkdtree{s} are generally faster than other baselines, since the object median based splitting yields better prune efficiency in the static case.
\hilberttree{s} are comparable to \ourorth{} on \knn{} queries, but significantly faster than \mortontree{s}, due to the better locality of Hilbert curve.
The performance of range-list queries are similar for all indexes, since the large output size hides the difference in pruning efficiency, so that the \hilberttree{s} and \mortontree{s} are more competitive with \pkdtree{s} and \porth{}.in this case.

\myparagraph{Dynamic updates}
We measure the time for incremental insertion and deletion with batch ratio 0.01\%, and perform queries after half of the batches.

\hilberttree{s} and \mortontree{s} are at least xx$\times$ than the other fastest baselines in both \cosmo{} and \osm{}, except on the \geolife{}, where it is xx$\times$ slower than \ourorth{} on insertion and xx$\times$ slower on deletion.
The reason is that compared with the \cosmo{} and the \osm{}, the \geolife{} is more uniformly distributed, and the dataset is much smaller, where the \ourorth{} has the smallest tree height meanwhile keep good balance.
The \pkdtree{} is slower on updates in skewed datasets, especially on the \cosmo{}, where the dense clusters make the tree easily unbalanced and requires frequent rebalancing.
The \ourorth{} benefits from the small coordinate range ($[0, 10^6]$), where it can keep a small tree height and still achieve good query performance on \cosmo{} and \osm{}, while our previous experiments in \cref{xxx} has shown that the \ourorth{} is less ideal for skewed data such as the \varden{} with large data range ($[0, 10^9]$).
}

\subsection{Scalability\label{sec:exp-scalability}}
We evaluate the scalability of tested indexes in construction, insertion, and deletions, which is illustrated in \cref{fig:scalability}.
We use $10^9$ 2D points. A batch insertion uses a single batch of size $10^7$.
The data points in \cref{fig:scalability} show the speedup over the 1-core performance of \hilberttree{},
and therefore \cref{fig:scalability} also reflects the true efficiency comparison of each index (higher is better).
In general, all indexes scale well to 224 hyperthreads, and the difference mainly comes from the work (i.e., one-core performance).
Among them, \hilberttree{} has the best self-relative speedup, which is up to 82.9$\times$ in build and 80$\times$ in insertion.
This is likely due to its simple structure as a 1D search tree. Combining both low work and good scalability, \hilberttree{} has the best overall construction and update performance.
The scalability on batch deletion is similar to that of batch insertion, where the \hilberttree{} has the best scalability on 112 cores, which is 67$\times$ on \sweep{}, 37.4$\times$ on \varden{} and 68.4$\times$ on \uniform{}.
The \ourorth{} have good scalability on \uniform{}, but is slightly worse on \sweep{} and \varden{} due to the imbalanced tree.



\hide{
We evaluate the scalability of tested indexes in construction and insertion in \cref{fig:scalability} (deletion is shown in\iffullversion{~\cref{app:scalability}}\ifconference{ the full version paper~\cite{psiPaperFull}}).
We use $10^9$ 2D points. A batch insertion uses a single batch of size $10^7$.
The data points in \cref{fig:scalability} show the speedup over the 1-core performance of \hilberttree{},
and therefore \cref{fig:scalability} also reflects the true efficiency comparison of each index (higher is better).
In general, all indexes scale well to 224 hyperthreads, and the difference mainly comes from the work (i.e., one-core performance).
Among them, \hilberttree{} has the best self-relative speedup, which is up to 82.9$\times$ in build and 80$\times$ in insertion.
This is likely due to its simple structure as a 1D search tree. Combining both low work and good scalability, \hilberttree{} has the best overall construction and update performance.}


\subsection{Summary} \label{sec:exp:summary}
Combining all the experimental results, we visualize the tradeoff between update and query performance for all tested indexes in \cref{fig:summary} and provide a brief summary here.
A more detailed explanation is illustrated in \cref{table:related-work}.

\myparagraph{\pkdtree{s}.} 
\pkdtree{s} offer solid performance in queries, but can degrade on \ods{} queries.
Its update performance is reasonable but less competitive than \porth{s} and \cpamtree{s}.
Our results suggest they are best suited to scenarios with light to moderate update rates, high query throughput requirements, and predominantly in-distribution queries.

\myparagraph{\porth{s} (this paper).} \porth{s} generally give the best overall performance and trade-off between query and updates, especially non-skewed data.
It is best suited to scenarios with less skewed data, with any update-query ratio. It is also friendly to queries after high-volumes of updates, since the tree quality does not degrade with frequent updates.

\myparagraph{\cpamtree{s} (this paper).} \hilberttree{} performs slightly worse in updates than \mortontree{}, but significantly better in queries.
We would recommend \hilberttree{} as the default setting for \cpamtree{s}.
Compared to \pkdtree{s} and \porth{s}, \hilberttree{s} are less effective in queries, but significantly faster in construction and updates.
In general, \hilberttree{s} are best suited to highly dynamic scenarios where either updates requires very high throughput/low latency, or updates are much more frequent than queries.


\hide{
\subsubsection{Query Performance}

Column ``Query after Build'' in \cref{fig:heatmap} shows query times on a tree built from $5\times 10^8$ points.
Most indexes are nearly perfectly balanced after construction, and thus these results reflect their best-case (static) query performance.
In \cref{fig:heatmap}, we only select results for 10-NN query and a relatively large range query.
To give more details, in \cref{fig:knn} and \cref{fig:rr}, we further show how query performance changes with the output size, i.e., the parameter $k$ in \knn{},
and the range size in range-list queries.

\begin{figure}[t]
	\includegraphics[width=0.48\textwidth]{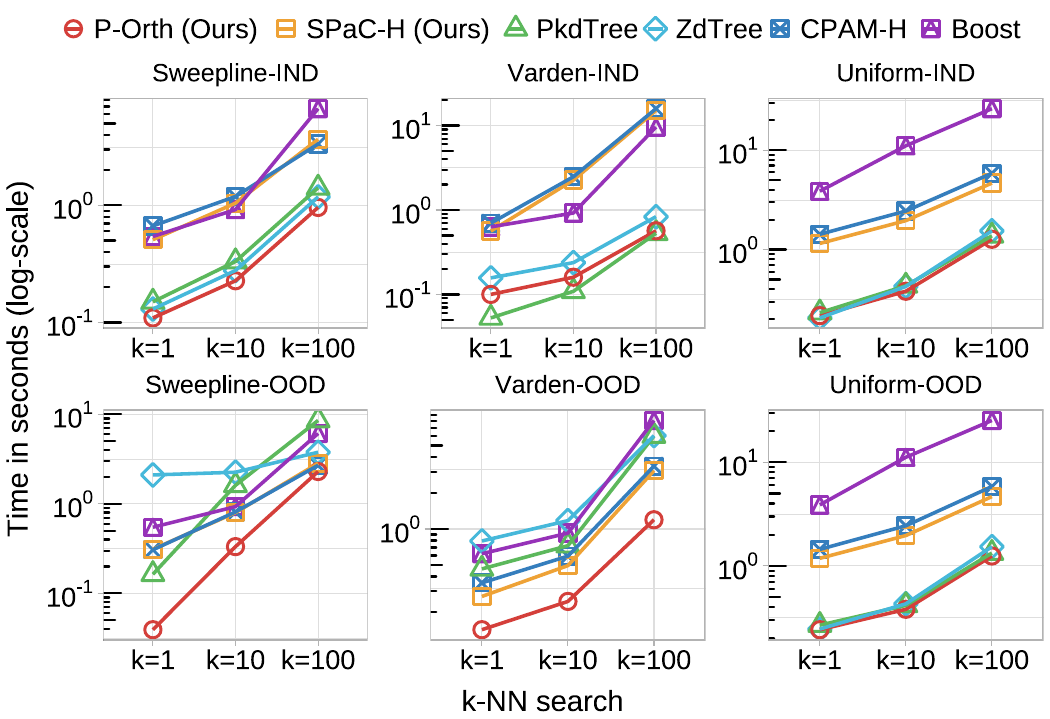}
	\caption{\textbf{Running time (in seconds) of $k$-NN queries for $k\in\{1,10,100\}$. Lower is better.}
		The dataset contains 500M points in 2 dimensions.
		The tree is constructed by incremental insertion with batch ratio $0.01\%$.
		The test contains $k$-NN queries from $10^7$ points from both \ids{} and \ods{} distribution.
		Plots are in log-log scale.
	}
	\label{fig:knn}
	\includegraphics[width=0.48\textwidth]{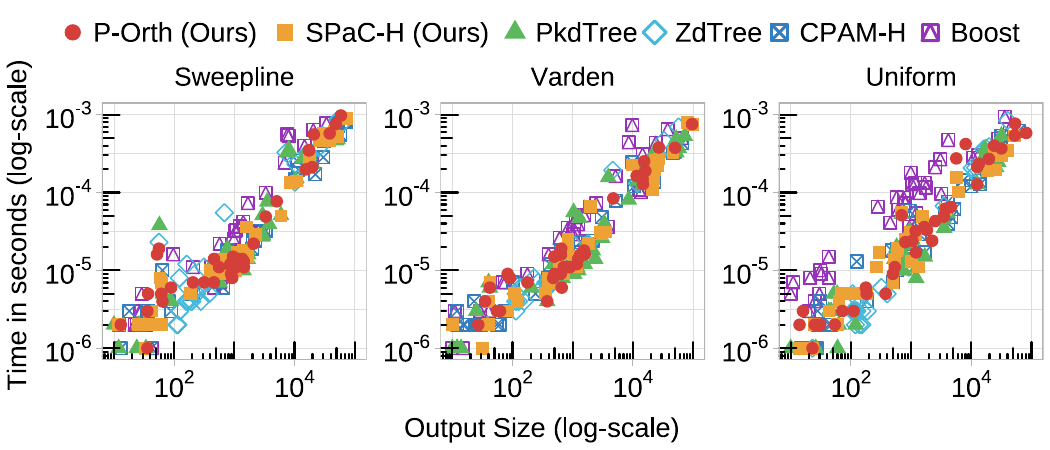}
	\caption{\textbf{Running time (in seconds) of range report queries for w.r.t output sizes. Lower is better.}
		The dataset contains 500M points in 2 dimensions.
		The tree is constructed by incremental insertion with batch ratio $0.01\%$.
		Plots are in log-log scale.
	}
	\label{fig:rr}
\end{figure}

\myparagraph{\knn{} Queries.}
As shown in \cref{fig:knn}, spatial-partitioning trees are evidently faster than R-trees in \knn{} queries, due to non-overlapping bounding boxes.
For \ourcpam{s}, while the \hilberttree{} is slightly slower than \mortontree{} in construction and updates, it is much more efficient in queries.
This is because Hilbert curve has better locality than Morton curve (adjacent codes are always geometrically close to each other).

Among the space-partitioning trees, \orthtree{s} has the best overall performance.
This is because when traversing the subtrees, the \ourorth{} can select 1 out of 4 quadrants,
which is more effective than \pkd{s} and \zd{s} that select 1 out of 2 half spaces.
Hence, \pkdtree{s} and \zdtree{s} are competitive but usually slower than \porth{s}.
The only exception is on \varden{} data.
For \ids{} queries, due to the skewed distribution of \varden{}, \orthtree{s} may be unbalanced,
and thus the comparison-based \pkdtree{s} perform better.
Interestingly, on the contrary, both \orthtree{s} exhibit an advantage on \ods{} queries on \varden{}.
The reason is still in imbalance---for \varden{}, points are highly clustered, making these regions in the tree deep and other regions shallow.
Since the \ods{} queries distribute differently from the input, they likely hit the shallow regions and thus are much faster.

\myparagraph{Range Queries.} As shown in both \cref{fig:heatmap,fig:rr},
\pkdtree{s} show a small but consistent advantage on range queries.
This is because a range query needs to visit all subtrees overlapping the query box.
In this case, \ourorth{s} have to explicitly check the bounding boxes for four subtrees,
while every non-overlapping checking on a \pkd{} node can prune half of the points in this subtree.
For other indexes, the relative performance on range queries are similar to \knn{} queries.
Interestingly, while \ourcpam{s} are still slower than \kdtree{s} and \porth{s} in range-list queries, the difference is much smaller, especially on large ranges---
in this case, the query time is mostly spent emitting the result list, hiding the difference in pruning effectiveness across indexes.
Therefore, compared with the \knn{} queries, range queries are less sensitive to the index type.


In summary, for queries on static trees,
\orth{s} and \kdtree{s} are naturally better than R-trees.
\kdtree{s} are better in dealing with \ids{} queries on non-uniform data, but may be worse in \ods{} queries.
\ourorth{} has the best or close to the best query performance in almost all queries and workloads.

\subsubsection{Queries on dynamic trees}
In addition to queries after a direct build,
we also run queries after 50\% of the batch insertion/deletions with batch size 0.01\% of the dataset.
These results reflect how the quality of each index is affected by updates.
The \orthtree{s} (\porth{} and \zdtree{}) are history-independent (modulo leaf-wrapping), namely,
the final state of the tree is not affected the order of operations applied to the tree.
Therefore, their performance is least affected by batch updates, making them perform best in queries in the dynamic setting.

For all other indexes, the tree 
may get less balanced after updates. Indeed, their performance all degrades to some extend compared to the static setting.
Similar to the static setting, \ourcpam{s} are less competitive in queries than \kdtree{s} and \orth{s},
which is common for R-tree structures due to overlapping bounding boxes.
However, compared to Boost R-tree, \ourcpam{s} achieve similar or better performance---in all queries,
\hilberttree{} is between 3.7x slower to 5.66x faster than Boost R-tree,
with a geometric mean of 2.5x faster.
}

\section{Conclusion\label{sec:conclusion}}

In this paper, we systematically study parallel spatial indexes, with a special focus on achieving high-performance updates in highly dynamic workloads. 
We proposed two new data structures: a parallel \orth{}, the \porth{}, and a parallel R-tree, the \cpamtree{} family. 
Both achieve superior update performance compared to existing parallel spatial indexes, while remaining competitive with or better than their counterparts in the literature for queries. 
We also highlight our comprehensive experiments to understand the performance of existing and our new parallel spatial indexes, and share our findings in \cref{sec:exp:summary} and \cref{fig:summary}. 

\section{Acknowledgements}
This work is supported by NSF grants CCF-2103483, TI-2346223 and IIS-2227669, NSF CAREER Awards CCF-2238358 and CCF-2339310, and the Google Research Scholar Program.


\ifconference{\clearpage}
\bibliographystyle{ACM-Reference-Format}
\balance
\bibliography{bib/strings, bib/main}

\ifconference{
  \appendix
  \section{Artifact Description}
\subsection{Abstract}
PSI (Parallel Spatial Indexes), also known the \ourlib{}, is a highly optimized \texttt{C++} library for parallel spatial indexing and querying on multi-core architectures.
It provides efficient implementations of three spatial index structures: the \kdtree{}, the \orthtree{}, and the \rtree{}, optimized for parallel processing.
\ourlib{} is designed to handle large-scale spatial datasets, enabling efficient batch operations (build/insertion/deletion) and fast spatial queries, such as nearest-neighbor and range searches.
The library provides both highly efficient implementation and strong theoretical guarantees, making it suitable for applications in geographic information systems (GIS), computer graphics, and spatial data analysis.

\subsection{Framework and Implementation}
The library is a header-only \texttt{C++} library, making it easy to integrate into existing projects.
All basic geometric primitives and fundamental tree operations are implemented in the \texttt{psi::BaseTree}, allowing for easy extension to support additional geometric types and operations.

The system implements three primary spatial tree types: \texttt{psi::KdTree}, \texttt{psi::OrthTree}, and \texttt{psi::PTree} plus two baseline implementations: \texttt{ZdTree} and \texttt{Boost-R-tree} for performance comparison.

All tree types support parallel construction via fork-join parallelism (ParlayLib), batch updates including \textsc{Build()}, \textsc{BatchInsert()}, \textsc{BatchDelete()}, and \textsc{BatchDiff()}, and spatial queries such as \textsc{KNN()}, \textsc{RangeQuery()}, and \textsc{RangeCount()}.

The correctness of the implementations is verified through extensive unit tests using a hand-crafted framework.
The library also shipped a parallel synthetic dataset generator to facilitate testing and benchmarking.

\subsection{Artifact Meta Information}
\begin{itemize}
  \item \textbf{Program}: \texttt{C++} library with parallisim supported by ParlayLib~\cite{blelloch2020parlaylib}.
  \item \textbf{Compilation}: \texttt{g++} 14.2.1 with \texttt{-O3} optimization flag.
  \item \textbf{Dataset}: Both the internally generated synthetic datasets and real-world datasets are used in the experiments.
  \item \textbf{Run-time environment}: Linux-based OS (e.g., Ubuntu 20.04).
  \item \textbf{Hardware}: Any x86-based multi-core machine with at least 512 GB RAM.
  \item \textbf{Output}: Performance metrics including construction time, batch update time, and query time.
  \item \textbf{Workflow}: Download the code and run the provided scripts to reproduce the experiments within a Docker container.
  \item \textbf{Publicly available}: Yes, the code is available at~\cite{psiCode}.
\end{itemize}

\subsection{Dependencies}
\myparagraph{Hardware.}
The experiments require a multi-core modern (2010+) x86-based multicore machine.
Relies on 128-bit CMPXCHG (requires -mcx16 compiler flag) but does not need hardware transactional memory (TSX).
The machine should have at least 512GB of RAM to handle large datasets used in the experiments.
We use a machine with 112 cores (224 hyperthreads), running Rocky Linux 8.10, with four Intel Xeon Platinum 8176 CPUs and 1.47 TB of RAM.

\myparagraph{Software.}
The code is written in \texttt{C++20} and can be compiled with \texttt{g++} version 14.2.1 or higher using the \texttt{-O3} optimization flag.
The external dependencies is included as a submodule in the repository.
The experiments are performed with \texttt{numactl} to enable better parallelism and \texttt{jemalloc} as the memory allocator to improve memory allocation performance.
\texttt{Docker} is recommended for quickly setting up the environment using the scripts in the repository.

\myparagraph{Datasets.}
The library shipped with a parallel synthetic dataset generator that can generate large-scale datasets with two distributions, \uniform{} and \varden{} (see \cref{sec:exp-synthetic} for more details about the distribution) and arbitrary dimensions.
Real-world datasets used in the experiments can be downloaded from public sources, with links provided in the documentation.

\subsection{Installation}
The code can be cloned from the public repository~\cite{psiCode}.
We provide a \texttt{Dockerfile} in the repository to quickly set up the environment.

\subsection{Experiments Workflow}
The artifact is easy to use as follows:
\begin{enumerate}
  \item Clone the repository from~\cite{psiCode}.
  \item Run the \texttt{./docker-run.sh pull} to build Docker image.
  \item Run the \texttt{./docker-run.sh full {-}{-}data-path [path]} to start a Docker container with the provided disk path for datasets, then all the experiments in the paper will run automatically.
  \item Check the results in the \texttt{script\_ae/plots} folder in the repository.
\end{enumerate}
All the dependencies are pre-installed in the Docker image, and the datasets will be generated or downloaded to the provided disk path.

If the \texttt{Docker} is not available, the code can also be compiled following the instructions in the \texttt{README.md}, and the experiments can be run using \texttt{script\_ae/run.sh}.

\subsection{Evaluation and Expected Results}
The running time for all operations (construction, batch update, and queries) will be printed to the standard output and saved in the \texttt{script\_ae/logs} folder.
All the plots in the paper will be generated and saved in the \texttt{script\_ae/plots} folder.
In our machine, execution times exhibit minimal variation across runs, with deviations of 1-3\%.
Note that the performance may vary on different hardware configurations, such as the number of cores. Our machine has 112 cores (224 hyperthreads), and we recommend using a machine with at least 64 cores to observe similar performance trends.

}

\iffullversion{
  \clearpage
  \appendix


\newcolumntype{L}[1]{>{\raggedright\arraybackslash\hspace{0pt}}m{#1}}


\clearpage
\begin{table*}[!t]
  \small
  \centering~\\~\\~\vspace{1em}
  \begin{tabular}{L{0.4cm}|L{6cm}|L{10cm}}
    \toprule
    & \multicolumn{1}{c|}{\textbf{General Features}}                                                                                                                                                                          & \multicolumn{1}{c}{\textbf{Existing Solutions and Their Features}}                                                                                                                                                                                                                                                                                   \\
    \midrule
    \begin{sideways}\textbf{\kdtree{}}
    \end{sideways}                         &
    \vspace{-.5cm}
    \begin{itemize}[noitemsep]
      \item[$+$] Linear space
      \item[$+$] Flexible for most queries (e.g., k-NN, range)
      \item[$+$] Non-overlapping bounding boxes (thus effective pruning in queries)
      \item[$+$] Generally fast queries across distributions
      \item[$+$] Comparison-based, resistant to skewed data
      \item[$+$] Easily generalizable beyond three dimensions
      \item[$-$] Slow/complicated updates
    \end{itemize}
    & \textbf{\pkdtree{}~\cite{men2025parallel}}
    \begin{itemize}[noitemsep]
      \item[$+$] I/O optimizations for construction and updates
      \item[$+$] Fast construction: $O(n\log n)$ work and polylogarithmic span
      \item[$+$] Among the fastest for queries in most tests, except for \ods{} queries on skewed distributions
      \item[$-$] $O(m\log^2 n)$ work for batch update of batch size $m$, unfriendly to workloads with frequent updates
    \end{itemize}
    \textbf{\bhltree{} and \logtree{}~\cite{yesantharao2021parallel}}
    \begin{itemize}[noitemsep]
      \item[$+$] Can leverage vEB layouts for query optimization (due to their static nature)
      \item[$+$] Construction with $O(n\log n)$ work and polylogarithmic span
      \item[$-$] Large batch-update cost: $O(m\log^2n)$ (\logtree{}) or $O((n+m) \log (n+m))$ (\bhltree{}, due to fully rebuild)
      \item[$-$] \logtree{} uses logarithmic method, leading to inefficient queries\vspace{-1em}
    \end{itemize}\\
    \midrule
    \vspace{1cm}
    \multirow{2}[1]{*}{
      \begin{sideways}\textbf{\orthtree{}}
    \end{sideways}} &
    \vspace{-.5cm}
    \begin{itemize}[noitemsep]
      \item[$+$] Linear space
      \item[$+$] Flexible for most queries (e.g., k-NN, range)
      \item[$+$] Non-overlapping bounding boxes (thus effective pruning in queries)
      \item[$+$] Fast queries, especially on non-skewed data
      \item[$+$] History-independent (modulo leaf wraps)
      \item[$+$] Simple/fast construction and updates, especially on non-skewed data
      \item[$-$] Sensitive to skewed data
      \item[$-$] Usually not generalizable beyond three dimensions
    \end{itemize}
    & \textbf{\ourorth{} (this paper)}
    \begin{itemize}[noitemsep]
      \item[$+$] I/O optimizations for construction and updates
      \item[$\star$] Fastest query performance on non-skewed data
      \item[$\star$] Usually faster updates than Pkd-trees, even on reasonably skewed data; slower than \cpamtree{s}
      \item[$+$] Fast construction: $O(n\log \Delta)$ work and polylogarithmic span
      \item[$\star$] Fast batch updates: $O(m\log \Delta)$ work and polylogarithmic span
      \item[$-$] Most affected by skewed data in construction, updates and queries; less efficient for \ids{} queries on skewed data
    \end{itemize}
    \textbf{\zdtree{}~\cite{blelloch2022parallel}}
    \begin{itemize}[noitemsep]
      \item[$+$] Relatively skew-resistant due to comparison sorting
      \item[$+$] Fast construction: $O(n\log n)$ work and polylogarithmic span
      \item[$+$] $O(m\log \Delta)$ work for batch update, where $\Delta$ is the aspect ratio
      \item[$-$] Generally slower updates/construction than the \porth{}
      \item[$-$] Integer coordinates and Morton curve only\vspace{-1em}
    \end{itemize}\\
    \midrule
    \vspace{1cm}
    \multirow{2}[4]{*}{
      \begin{sideways}
        \textbf{R-tree/BVH}
    \end{sideways}}    &
    \vspace{-.5cm}
    \begin{itemize}[noitemsep]
      \item[$+$] Linear Space
      \item[$+$] Flexible rules due to object-partitioning
      \item[$+$] Simple/fast construction and updates
      \item[$+$] Applicable to common queries (e.g., k-NN, range)
      \item[$+$] Easily generalizable beyond three dimensions
      \item[$-$] Overlapping bounding boxes (thus ineffective pruning in queries); usually slower queries than space-partitioning trees
    \end{itemize}
    & \textbf{\spactree{} (this paper)}
    \begin{itemize}[noitemsep]
      \item[$\star$] Compatible with Hilbert, Morton or other space-filling curves
      \item[$\star$] Embeds multi-dimensional data to 1D, enabling simple algorithm design and high parallelism (best self-speedup among tested indexes)
      \item[$+$] Fast construction: $O(n\log n)$ work and polylogarithmic span
      \item[$\star$] Super fast batch updates: $O(m\log n)$ work and polylogarithmic span
      \item[$+$] Comparison-based; robust to skewed data
      \item[$\star$] Fastest construction and update time among all baselines; significant advantage on updates
      \item[$-$] Integer coordinates only
      \item[$-$] Slow queries than space-partitioning trees due to overlapping bounding boxes
    \end{itemize}
    \textbf{Boost R-tree~\cite{schaling2011boost}}
    \begin{itemize}[noitemsep]
      \item[$+$] Supports multiple heuristics
      \item[$-$] No parallel construction or batch updates
      \item[$-$] Slow queries than space-partitioning trees due to overlapping bounding boxes\vspace{-1em}
    \end{itemize}
    \\
    \midrule
    \begin{sideways}\textbf{Range tree}
    \end{sideways}                        &
    \vspace{-.5cm}
    \begin{itemize}[noitemsep]
      \item[$+$] Worst-case work bound for range queries
      \item[$-$] $O(n\log n)$ space
      \item[$-$] Only supports range queries
      \item[$-$] Inefficient in more than two dimensions\vspace{-1em}
    \end{itemize}
    & \textbf{CPAM/PAM range tree~\cite{sun2018pam,dhulipala2022pac}}
    \begin{itemize}[noitemsep]
      \item[$+$] Parallel construction with $O(n\log n)$ work and polylogarithmic span
      \item[$-$] No simple support for parallel batch updates
    \end{itemize}    \\
    \bottomrule
  \end{tabular}%

  \caption{Summary of the main features of different spatial trees and existing solutions for parallel construction, updates, and queries. The symbol ``$\star$''marks our key technical contributions.
    In the bounds, $m$ is the batch size, $n$ is the index size, and $\Delta$ is the aspect ratio.
    \label{table:related-work}
  }
\end{table*}

\clearpage

\hide{
  \begin{table*}[]
    \small
    \centering
    \begin{tabular}{L{0.4cm}|L{6cm}|L{10cm}}
      \toprule
      & \multicolumn{1}{c|}{\textbf{General Features}}                                                                                                                                                                          & \multicolumn{1}{c}{\textbf{Existing solutions and their features}}                                                                                                                                                                                                                                                                                   \\
      \midrule
      \begin{sideways}\kdtree{}
      \end{sideways}                         & (+) Linear space\newline{}(+) Fast query\newline{}(+) Non-overlapping bounding boxes\newline{}(+) Purely comparison-based\newline{}(+) Flexible to most queries\newline{}(-) Slow/complicated update & \pkdtree{}\newline{}(+) I/O optimizations for construction and update. \newline{}(+) Construction with $O(n\log n)$ work and polylog span\newline{}(+) Fastest or close to fastest query performance in most tests, except for \ods{} query\newline{}(-) $O(m\log^2 n)$ work for batch update on batch size m\newline{}(-) Unfriendly to skewed data \\
      \midrule
      \multirow{2}[4]{*}{
        \begin{sideways}Oct/Quad tree
      \end{sideways}} & \multicolumn{1}{l|}{\multirow{2}[4]{*}{\makecell[l]{(+) Linear space                                                                                                                                                                                                                                                                                                                                                                                                                                                                                                           \\(+) Non-overlapping bounding boxes\\(+) Flexible to most queries\\(+) History-independent\\(+) Fast construction and update\\(-) Not resistant to skewed data}}} & \ourorth{} (this paper)\newline{}(+) I/O optimizations for update/construction\newline{}(++) Faster update than Pkdtrees, with fastest or close to fastest query performance on uniform data\newline{}(++) Optimal work, span and I/O for both tree construction and batch update\newline{}(-) Update and query performance may be slower than pkdtrees on skewed data \\
      \cmidrule{3-3}                                                  & \multicolumn{1}{l|}{}                                                                                                                                                                                                   & \zd{}\newline{}(+) Known extension to concurrent setting~\cite{xxx}\newline{}(-) Slower update/construction than p-orth tree\newline{}(-) Only work on integer keys                                                                                                                                                                                  \\
      \midrule
      \multirow{2}[4]{*}{
        \begin{sideways}R-tree/BVH
      \end{sideways}}    & \multicolumn{1}{l|}{\multirow{2}[4]{*}{\makecell[l]{(+) Flexible spatial partition rule                                                                                                                                                                                                                                                                                                                                                                                                                                                                                        \\(+) Linear space\\(+) Flexible to most queries\\(+) Easy and fast update\\(-) Overlapping of bounding boxes}}} & \boldmath{}\spactree{} (this paper)\newline{}(+) Compatible with hilbert, morton or other space filling curves\newline{}(+) Maps multi-dimensional data to 1D, allows for simple algorithm design\newline{}(+) Comparison-based, resistant to skewed data\newline{}(++) Fastest construction time and update time among all baselines\newline{}(++) $O(m\log n)$ batch update cost for batch size m and tree size n\newline{}(-) Slowest query performance due to overlapping bounding boxes\unboldmath{} \\
      \cmidrule{3-3}                                                  & \multicolumn{1}{l|}{}                                                                                                                                                                                                   & Boost R-tree\newline{}(+) Various heurestics supported\newline{}(-) No parallel construction or batch updates                                                                                                                                                                                                                                        \\
      \midrule
      \begin{sideways}Range-tree
      \end{sideways}                        & (+) Worst-case work bound for 2D range queries\newline{}(-) $O(n\log n)$ space\newline{}(-) Only support range queries\newline{}(-) Unfriendly to more than 2D                                                          & \boldmath{}CPAM/PAM range tree:\newline{}(+) Parallel construction with $O(n\log n)$ work and $O(\log^3 n)$ span\newline{}(-) No direct support for parallel batch updates\unboldmath{}                                                                                                                                                              \\
      \bottomrule
    \end{tabular}%

    \label{table:related-work}%
    \caption{A summary of the main features of different spatial trees, and existing solutions for parallel construction, update and query.}

  \end{table*}

}

\section{Analysis on \titlecap{\porth{s}}\label{app:porth-analysis}}
We now analyze the theoretical guarantees for our \ourorth{} construction and batch update algorithms.
Let $S\subseteq M$ a finite point set in the bounded Euclidean space $M$ and denote $B_p(r)\subseteq S$ the set of points enclosed by a ball with radius $r$ centered at $p$. Then $S$ has $(\rho,c)-$\bdita{expansion} if and only if $\forall p\in M$ and $r>0$:
\begin{equation}
  |B_p(r)|\geq \rho \implies |B_p(2r)|\leq c\cdot|B_p(r)|
\end{equation}
The constant $c$ is referred to \bdita{expansion rate} and $\rho$ is usually set to be $O(\log |S|)$. We say the expansion rate is \textit{low} if $c=O(1)$.
Intuitively, the low expansion property ensures the points distributed uniformly in the space.

Similarly, the \bdita{aspect ratio} $\Delta$ is defined as:
\begin{equation}
  \Delta = \frac{\max d(x,y)}{\min d(x,y)}\quad \forall x,y\in S
\end{equation}
and is said to be \textit{bounded} if $\Delta< n^c$ holds for some constant $c>0$.

We will now show that \ourorth{} with the assumption of bounded aspect ratio.
Without the assumptions, the tree height becomes $O(\log \Delta)$.
We can replace the tree heights in the following analysis to get \cref{thm:porth}.
\begin{lemma}\label{lem:orth_height}
  The height for \ourorth{} on points $P$ with size $n$ is $O(\log n)$, assuming the low expansion rate and the bounded aspect ratio for $P$.
\end{lemma}{}
\begin{proof}
  By the low expansion rate, the $H$ has side length at most a constant fraction of $d_{max}$, and the recursion stops when two points with distance $d_{min}$ are separated.
  Since $d_{max}/d_{min}=n^c$ by the bounded aspect ratio, and the splitters cut the $H$ in the spatial median, it takes $O(\log n)$ levels of splitters to reduce the side length of $H$ to $d_{min}$. The proof follows then.
\end{proof}

With the above lemma, we now show our \orth{} construction algorithm has $O(n\log n) $work, polylogarithmic span and $O(\mbox{Sort}(n))$ cache complexity on $n$ points.
Here wssume the cache size $M= \Omega(\text{polylog}(n))$ as in~\cite{blelloch2010low, BG2020}, by setting the skeleton height $\lambda=\epsilon \log(M)$ for $\epsilon<1/(2D)$, and chunk size $l=2^{D\cdot \lambda}$ in the sieving algorithm.
Denote $O(\mbox{Sort}(n)) = O(n/B\cdot\log_Mn)$ the optimal cache complexity for sorting~\cite{blelloch2010low}.

\begin{theorem}{}\label{thm:orth_constr}
  With parameters specified above, a \ourorth{} can be constructed on points $P$ with size $n$ in $O(n\log n)$ work, $O(\log^2 n)$ span and $O(\mbox{Sort}(n))$ cache complexity, assuming the low expansion rate and the bounded aspect ratio for $P$.
\end{theorem}{}
\begin{proof}
  Every point is processed at most once in each round, except for the points sieving, where finding the bucket for one point takes $O(\skheight\cdot D)$ work. The algorithm terminates after $O(\log n)/\skheight$ rounds of recursion, which implies $O(\lambda\cdot D)\cdot O(\log n)/\lambda=O(\log n)$ total work per point. Therefore, the total work is $O(n\log n)$.

  For the span, practically the tree skeleton construction and processing each block is done sequentially.
  However, theoretically, they can be parallelized in $O(\lambda\log n)$ and $O(\log n)$ span, respectively~\cite{blelloch2010low}.
  All other operations takes $O(\log n)$ span.
  In total, the span in each round is $O(\lambda\log n)$. The algorithm has $O(\log n)/\skheight$ rounds of recursion, so the overall span is $O(\log^2 n)$.

  Now consider the cache complexity. Both building the tree skeleton and sub-regions computation fully fit in cache. The chunk size $l=2^{\lambda\cdot D} = M^{\epsilon\cdot D}\leq\sqrt{M}$, which implies that each chunk fully fits in cache.
  Therefore, the sieving algorithm takes $O(n/B)$ block transfers.
  All other operations take $O(n/B)$ block transfers, in total $O(n/B\cdot \log n/\lambda)=O(n/B\cdot\log_M n)$ I/Os.
\end{proof}

For batch updates, we assumes the batch size $m=O(n)$, and if $m=\omega(n)$, we simply replace $n$ with $m+n$ in below bounds for insertions, and there is no change for deletions.
\begin{theorem}{}
  The Update (insertion or deletion) of a batch of size $m=O(n)$ on a \ourorth{} of size $n$ can be performed in optimal $O(m\log n)$ work, $O(\log^2 n)$ span, and $O(m(\log (n/m)+(1+\log_M m)/B))$ cache complexity, assuming the low expansion rate and the bounded aspect ratio for the updated points.
\end{theorem}{}
\begin{proof}
  We take the insertion as an example, the deletion is similar.
  For the work, note the tree after updates is same as the one built from scratch on all points, which has height $O(\log (m+n))$ by Lem. \ref{lem:orth_height}. The height difference is $O(\log (m+n))-O(\log n)=O(1)$. Since each leaf wraps $O(1)$ points, and every point needs $O(\log n)$ work to reach the leaf, the total work is $O(m\log n)$.

  The analysis for span is the same as for construction in \cref{thm:orth_constr}.

  The cache bound for updates has two parts.
  The first is sorting within the batch.  This part has $O(m(1+\log_M m)/B)$ cache complexity.
  The second part is accessing the tree nodes in the original \porth{}.
  Finding $m$ leaves in a tree of size $n$ will touch $O(m\log(n/m))$ tree nodes~\cite{blelloch2016just}.
  Putting both cost together gives the stated cache complexity.
\end{proof}

Replacing all the tree height $O(\log n)$ with $O(\log \Delta)$ gives \cref{thm:porth}, without the assumption of bounded aspect ratio.

\section{Analysis on \titlecap{\spactree{s}}}\label{app:spac-analysis}

\subsection{Correctness}

We prove the correctness of the update algorithms for \spactree{s} by showing its equivalence to that of the \pactree{}.
Here we discuss the insertion algorithm, and deletion can be shown similarly.
First, the \cref{algo:ptree-constr} constructs the same tree as the \pactree{},
so the tree returned in line~\ref{line:insert-base-case} and line~\ref{line:insert-rebuild} remains the same.
The split key in line~\ref{line:insert-getkey} is same for both trees, therefore the \spactree{} will insert same points in leaves as the \pactree{} in line~\ref{line:cpam_insert_append}, but keep the points unsorted.
The \join{} and \rightJoin{} operations (line~\ref{line:insert-join} and line~\ref{line:right-join}) are identical for both tree. In this case, the tree split will reach the same leaves in both trees, and line~\ref{line:reorder-in-expose} and line~\ref{line:sort-in-node} ensure the points order in \spactree{} to be identical to those in \pactree{} before further proceeding.
The other operations in \expose{} and \node{} remains the same, and the correctness follows.

\subsection{Cost Analysis}

\begin{theorem}{}\label{thm:spac_constr_app}
  A \spactree{} with $n$ points can be constructed in $O(n\log n)$ work, $O(\log n)$ span, and $O(\mbox{Sort}(n))$ cache complexity.
\end{theorem}
\begin{proof}
  The \sortSfc{} in\cref{algo:ptree-constr} is a simple modification of the sample-sort algorithm~\cite{blelloch2010low, axtmann2017place}---all extra operations (i.e., computing the SFC code and storing the point id) take no additional asymptotic cost.
  The \cpamBuildFromSorted{} in \cref{algo:ptree-constr} is a parallel divide-and-conquer algorithm that takes $O(n)$ work, $O(\log n)$ span and $O(n/B)$ cache complexity.
  The proof then follows.
\end{proof}

In the following proof we assume the batch size $m=O(n)$.
Note that the following proof depends on the analysis of \pactree{}, which can be found at~\cite{dhulipala2022pac}.
\begin{theorem}
  A batch update (insertion or deletion) of size $m$ on a \spactree{} of size $n$ uses $O(m\log n)$ work and $O(\log^2 n)$ span.
\end{theorem}
\begin{proof}
  We first show the span bound. The sorting takes $O(\log m)$ span, and the following points insertion/deletion takes $O(\log n)$ rounds to reach leaves.
  Expanding the leaf and restore the points order take constant time. Both the \join{} and \rightJoin{} take $O(\log m)$ span~\cite{dhulipala2022pac} in each round. In total, the entire process has $O(\log m\log n)$ span.

  We now show the work bound.
  Sorting $m$ points takes $O(m\log m)$ work, and each point takes $O(1)$ operation in each round. The leaf expansion takes constant time. The work by \join{} is asymptotically bounded by the work of \rightJoin{}~\cite{dhulipala2022pac}, and the total work of \rightJoin{} is $O(m\log \frac{n}{m})$. The total work therefore is $O(m\log n)$.
\end{proof}

Combining both lemmas gives \cref{thm:spac}.

\section{Implementation Details}\label{app:impl}
\myparagraph{\ourorth{}.}
We choose to build $\skheight=3$ levels for 2D points and $\skheight=2$ levels for 3D points in the \ourorth{} skeleton, which provides generally good performance on our machine.
For each bounding box, we store only the point coordinates (with no extra metadata) of the lower-left and upper-right corners to save memory.
A single \knn{} query traverse subtrees in increasing order of their minimum distance to the query point, computed by comparing the query point with the bounding box associated with each subtree.

\myparagraph{\spactree{}.}
The implementation of \spactree{} builds on the code of \pactree{}~\cite{dhulipala2022pac}, but with a careful redesign to optimize performance.
Simply treating \pactree{} as a black box introduces overhead from transforming input points into key-value pairs---using the SFC code as the key and the entire point as the value---as suggested in its user manual.
We avoid this by redesigning the interface so that \spactree{} automatically parses the SFC code in each point as the key and treats the remaining attributes as the value.
This allows it to operate directly on the input sequence, reducing preprocessing time and memory usage.

We also introduce a heuristic to optimize batch updates when a leaf node overflows.
The original approach in \cref{algo:cpaminsert} unconditionally rebuilds the parent subtree by invoking \cref{algo:ptree-constr}.
This can be inefficient when many points are affected, because the insertion batch must be merged with the points in the leaves prior to recursive node allocation, even when the batch is already sorted.
Hence, it incurs significant overhead.
An alternative is to explicitly expose the leaf as a balanced tree with empty leaves and then perform the batch insertion on that subtree.
Our method chooses between these strategies via a threshold.
If the combined size of the overflowing leaf and the new batch is below a threshold (in our case, $4\leafwrap$), we perform a standard, localized rebuild. Otherwise, we expose the leaf and perform batch insertion on the exposed subtree.

\hide{
  We also use a heuristics approach to optimize batch updates when a leaf node overflows. The original approach in~\cref{algo:cpaminsert} would be to unconditionally rebuild the parent subtree by invoking \cref{algo:ptree-constr}.
  However, this is inefficient if the number of affected points is large, since the insertion batch has to be merged with leaf points before the recursive node allocation, even if the batch is already in the sorted order, which incurs large overhead.
  A more efficient approach is to explicitly expose the leaf into a balanced tree with empty leaves, and perform the batch insertion on this new subtree.
  Our method instead uses a threshold to decide the appropriate action. If the combined size of the points in the overflowing leaf and the new batch is below a threshold (in our case, $4 \leafwrap$), we perform a standard, localized rebuild.
Otherwise, we expose the leaf and perform the batch insertion on the exposed subtree.}

\myparagraph{Parameter Choosing.}
We empirically set the parameters to achieve the best performance on our machine for all implementations.
We set the leaf wrap to 40 for both \ourcpam{} and \cpam{}, and 32 for all other baselines.
Both \ourcpam{} and \cpam{} use the weight-balanced scheme with balancing parameter $\alpha=0.2$, i.e., the weights of left and right sub-tree can be differ by at most 20\%.
For \pkd{}, we adopt $\alpha=0.3$ as suggested in their paper.

\begin{figure*}[t]
  \centering\vspace{-0.5em}
  \includegraphics[width=\textwidth]{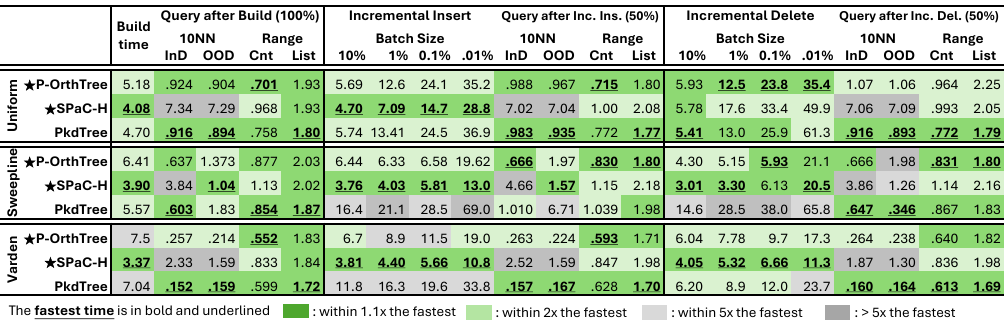}
  \caption{{\bf Running time (in seconds) on 3-dimensional synthetic data. Lower is better.} \normalfont
    The fastest time in each test is in bold and underlined.
    We use colors to mark results within 1.1$\times$, 2$\times$, 5$\times$, and $>5\times$ the fastest time.
    Detailed settings for build, queries, and incremental insertion/deletion are introduced at the beginning of \cref{sec:exp-synthetic}.
    \ids{}/\ods: in-/out-of-distribution.
  }\label{fig:heatmap-3d}
  \vspace{-0.7em}
\end{figure*}

\section{Batch Updates}\label{app:batch-updates}
We now provide more experimental results for single batch updates.
We evaluate the performance of batch updates by varying the batch size
from $10^5$ to $10^9$ points, with results presented in~\cref{fig:batch-update}. The
experimental setup consists of an initial tree constructed with $10^9$
points. We then perform two separate operations: a batch insertion, which
adds new points drawn from the same distribution, and a batch deletion,
which removes an equivalent number of existing points from the tree.
Smaller batch sizes were omitted from this analysis, as their low
computational cost diminishes the practical benefits of parallelism.

All baselines scale well on both single batch insertions and deletions.
The \hilberttree{} is faster than others on all benchmarks, except for the batch deletion on \uniform{}, where it is slightly slower than the \ourorth{} due to the handling of imbalance.
The \ourorth{} is slower than the \hilberttree{} on batch insertions on \varden{}, since it is skewed on highly clustering data, and the tree traversal time incurs more overhead.
The \pkdtree{} is generally slower than \hilberttree{} on skewed datasets such as \sweep{} and \varden{}, since its reconstruction-based balancing scheme is more expensive when the rebuilt sub-tree is large, which is typical on skewed datasets.

\begin{figure}[t]
  \includegraphics[width=0.48\textwidth]{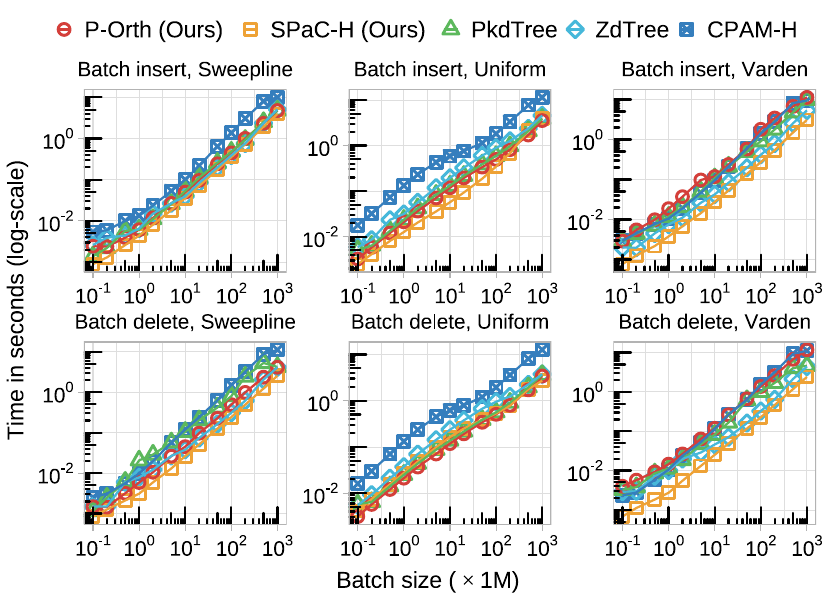}
  \caption{\textbf{
    Running time for batch updates on points from synthetic datasets on a tree with 1000M points in 2 dimensions.  Lower is better.}
    The batch insertion is to insert another batch from same distribution, and the batch deletion is to delete a batch from the existing points.
    The batch size is the number of points in the batch ($\times$ 1M).
    Plots are in log-log scale.
  }
  \label{fig:batch-update}
\end{figure}

\section{Performance on 3D Synthetic Datasets}\label{app:3d-syn}

We now provide more experimental results on 3D synthetic datasets.
We generate 3D synthetic datasets with the same method as in 2D, but limit the coordinates range within $[0, 10^{6}]$ to make it compatible to the \hilberttree{}.
The experiments set up is the same as in 2D~\cref{sec:exp}.
The results are shown in \cref{fig:heatmap-3d}.
We omit other baselines since they have been shown to be slower as in 2D case~\cref{fig:heatmap} .

For tree construction, \hilberttree{s} remain the fastest ones, and the time is more close to the 2D case compared with \ourorth{s} and \pkdtree{s},
since the SFC-based indexes are less sensitive to the dimensionality.
As a result, \hilberttree{s} are 1.3--2.2$\times$ faster than \ourorth{s} and 1.2--2.1$\times$ faster than \pkdtree{s}.

Regarding the batch updates, the \hilberttree{} remains the fastest one on most of the benchmarks, except on \uniform{} where it is slightly slower than the \ourorth{} due to handling of imbalance.
\ourorth{s} are faster than \pkdtree{s} on all benchmarks.
The reasons are 1) the \ourorth{} does not need to handle the imbalance, and 2) the range of coordinates is limited to $[0, 10^6]$, which enables the tree height of \ourorth{s} become much smaller than it in 2D cases (the data range is $[0, 10^9]$), so that the tree traversal time is much reduced in the skewed data such as \sweep{} and \varden{}.
However, \pkdtree{s} still keep the advantage on queries---the fastest one on most of the benchmarks, and competitive on the rest.

}
\end{document}
\endinput